\documentclass[11pt]{article}
\pdfmapfile{=cm.map}
\pdfmapfile{+cmextra.map}
\pdfmapfile{+symbols.map}
\PassOptionsToPackage{dvipsnames}{xcolor}
\usepackage{fullpage}

\usepackage{hyperref}
\usepackage[utf8]{inputenc} % allow utf-8 input
\usepackage{url}            % simple URL typesetting
\usepackage{mathtools}
\usepackage{booktabs}       % professional-quality tables
\usepackage{amsfonts}       % blackboard math symbols
\usepackage{nicefrac}       % compact symbols for 1/2, etc.
\usepackage{microtype}      % microtypography
\usepackage{xcolor}         % colors
\usepackage{xspace}

\usepackage{tikz}
\usetikzlibrary{decorations.pathreplacing}
\usepackage{pict2e,picture,graphicx}
\usepackage[overload]{empheq}
\usepackage{microtype}
\usepackage{graphicx}
\usepackage{subcaption}
\usepackage{booktabs}
\usepackage{amsmath}
\usepackage{cases}
\usepackage{mathtools}
\usepackage{amsthm}
\usepackage{thm-restate}
\usepackage{enumitem}

\allowdisplaybreaks
\usepackage{xparse}
\usepackage[capitalize, noabbrev]{cleveref}
\usepackage{wrapfig}
\usepackage{nicefrac} 
\usepackage{multirow}
\usepackage{bm}
\usepackage{subcaption}

\hypersetup{
	colorlinks = true,
	linkcolor = black,
	citecolor = black,
	linktocpage = true,
	urlcolor = darkGreen
}
\usepackage{comment}

\usepackage{amssymb}
\microtypesetup{expansion=false}
\theoremstyle{plain}
\newtheorem{theorem}{Theorem}[section]

\newtheorem{lemma}[theorem]{Lemma}

\newtheorem*{theorem-non}{Theorem}
\newtheorem*{corollary-non}{Corollary}

\theoremstyle{definition}

\theoremstyle{remark}

	\definecolor{niceRed}{HTML}{BE2626}
	\definecolor{Red2}{HTML}{DB3236}
	\definecolor{mgreen}{HTML}{A0C88C}
	\definecolor{blueGrotto}{HTML}{059DC0}
	\definecolor{limeGreen}{HTML}{81B622}
	\definecolor{myellow}{HTML}{E09B23}
	\definecolor{darkGreen}{HTML}{2E8B57}
	\definecolor{navyBlueP}{HTML}{03468F}
	\definecolor{Sepia}{HTML}{7F462C}
	\definecolor{orange2}{HTML}{FF8000}
	\definecolor{mgray}{HTML}{ABB3B8}
	\definecolor{lgray}{HTML}{E5E8E9}
	\definecolor{myPurple}{HTML}{AF007C}
	\definecolor{mypurple2}{HTML}{CC9EFF}
	\definecolor{royalBlue}{HTML}{057DCD}
	\definecolor{mpink}{HTML}{FC6C85}
	\definecolor{lblue}{HTML}{4A90E2}
	\definecolor{peagreen}{HTML}{98C127}
	\definecolor{typnavy}{HTML}{001F3F}
	\definecolor{typblue}{HTML}{0074D9}
	\definecolor{typaqua}{HTML}{7FDBFF}
	\definecolor{typteal}{HTML}{39CCCC}
	\definecolor{typeastern}{HTML}{239DAD}
	\definecolor{typpurple}{HTML}{B10DC9}
	\definecolor{typfuchsia}{HTML}{F012BE}
	\definecolor{typmaroon}{HTML}{85144B}
	\definecolor{typred}{HTML}{FF4136}
	\definecolor{typorange}{HTML}{FF851B}
	\definecolor{typyellow}{HTML}{FFDC00}
	\definecolor{typolive}{HTML}{3D9970}
	\definecolor{typgreen}{HTML}{2ECC40}
	\definecolor{typlime}{HTML}{01FF70}
	\definecolor{newgreen}{HTML}{83C702}
	\definecolor{mathchaYellow}{HTML}{F8E71C}
	\definecolor{mathchaGreen}{HTML}{7ED321}
	\definecolor{mathchaPurple}{HTML}{BD10E0}
	\definecolor{mathchaBlue}{HTML}{58B1FF}
	\definecolor{mathchaRed}{HTML}{FF5B59}
	\definecolor{metablue}{HTML}{0064E0}
\newcommand{\cD}{\mathcal{D}}

\usepackage{float}
\newfloat{algorithm}{tbp}{loa}
\floatname{algorithm}{Algorithm}
\usetikzlibrary{arrows.meta,calc,shapes.geometric}
\usepackage[round]{natbib}
\newcounter{algline}
\newcommand{\AlgLineNumber}{\stepcounter{algline}%
  \llap{\makebox[1.4em][l]{\normalfont\scriptsize\arabic{algline}}\hskip\leftskip}}
\newcommand{\KwIn}[1]{\par\noindent\textbf{Input:} #1\par
  \setcounter{algline}{0}\everypar{\AlgLineNumber}}
\newcommand{\For}[2]{\par\noindent\textbf{for} #1 \textbf{do}\par
  \begingroup\advance\leftskip by 1.6em\relax #2\par\endgroup}
\newcommand{\Repeat}[2]{\par\noindent\textbf{repeat}\par
  \begingroup\advance\leftskip by 1.6em\relax #2\par\endgroup
  \noindent\textbf{until} #1\par}
\newcommand{\eIf}[3]{\par\noindent\textbf{if} #1 \textbf{then}\par
  \begingroup\advance\leftskip by 1.6em\relax #2\par\endgroup
  \noindent\textbf{else}\par
  \begingroup\advance\leftskip by 1.6em\relax #3\par\endgroup}
\AtBeginEnvironment{algorithm}{\parindent=0pt\relax
  \let\originalsemicolon\;\renewcommand{\;}{\par}}

\DeclareMathOperator{\supp}{supp}
\newcommand{\1}{\mathbf{1}}
\newcommand{\E}{\mathbb{E}}

\newcommand{\Reg}{\mathrm{Reg}}

\title{Support-Dependent Regret Rates for Profit\\ Maximization in Multilateral Trade}
\author{\begin{tabular}{@{}c@{\hspace{2em}}c@{}}
  \begin{tabular}{c}Francesco Bacchiocchi\\Politecnico di Milano\end{tabular} &
  \begin{tabular}{c}Matteo Castiglioni\\Politecnico di Milano\end{tabular} \\[6mm]
  \begin{tabular}{c}Anna Lunghi\\Politecnico di Milano\end{tabular} &
  \begin{tabular}{c}Alberto Marchesi\\Politecnico di Milano\end{tabular}
\end{tabular}}
\date{}

\begin{document}
\maketitle

\begin{abstract}
We study online profit maximization in multilateral trade, a setting that generalizes both dynamic pricing and bilateral trade. At each round, an intermediary posts prices to buyers and offers payments to sellers; trade occurs only if every agent accepts.
We consider valuations drawn from an unknown joint distribution, allowing arbitrary correlation across agents, under the assumption that each agent's marginal valuation distribution has unknown support of size at most $K$. For constant dimension $d$, we design an algorithm with regret $\widetilde O(\sqrt{K^dT})$. This extends optimal support-dependent guarantees for one-dimensional dynamic pricing to the multidimensional setting of multilateral trade. At the core of our approach is a normalized zooming procedure that adapts to the local probability of trade, avoiding the excessive refinement of nearly flat regions that can arise with standard zooming.
\end{abstract}

\section{Introduction}

We study online profit maximization in multilateral trade, where an intermediary facilitates transactions between a set of buyers $B$ and a set of sellers $S$. At each round, the intermediary posts a price to each buyer and offers a payment to each seller. Trade occurs only if every buyer is willing to pay its posted price and every seller is willing to accept its offered payment. When trade occurs, the intermediary earns the difference between the total payments collected from the buyers and the total payments made to the sellers; otherwise, its profit is zero. Agents' valuations are drawn independently across rounds from an unknown joint distribution. The intermediary observes only whether the trade happens and seeks to minimize regret with respect to the best fixed vector of posted prices in hindsight.
This framework includes, as special cases, online pricing \citep{Kleinberg2003} with a single buyer and no sellers, and online bilateral trade \citep{cesa2021regret} with one buyer and one seller, and the joint-ads model \citep{aggarwal2024selling} with two buyers and no sellers.

Recent work has studied profit maximization in repeated trade under different benchmarks and feedback models. \citet{digregorio2025} consider the richer benchmark of all incentive-compatible and individually rational mechanisms, going beyond fixed posted prices. For bilateral trade with fixed prices and binary trade feedback, \citet{castiglioni2026sample} provide optimal sample complexity bounds and \citet{castiglioni2026breakingt34barrierregret} obtain $\widetilde O(T^{7/10})$ regret for arbitrary valuation distributions, leaving a gap with the $\Omega(T^{2/3})$ lower bound already present in pricing settings \citep{Kleinberg2003}. These results naturally raise the question of whether additional structure in the valuation distributions can lead to faster learning.

In this work, we focus on distributions in which each agent's valuation has finite support. Specifically, the valuation of agent $i$ is supported on an unknown set of size at most $K_i$; neither the locations of the support points nor their probabilities are known to the intermediary. Importantly, we impose no independence assumption across agents: the joint distribution may exhibit arbitrary correlations. We characterize the complexity of the problem depending on $K=\max_i K_i$, \emph{i.e.}, the largest support among all agents.

Our work builds on a line of research on dynamic pricing with finitely supported valuations, which corresponds to the special case of our model with a single buyer and no sellers. \cite{pmlr-v98-cesa-bianchi19a} initiated the study of the problem when the buyer's valuation distribution has an unknown support of size $K$. They established a regret bound $\widetilde O(\sqrt{KT})$, but with an additional instance-dependent term that can grow polynomially with the inverse distance between consecutive valuations. Recently, this instance-dependent term was removed, independently and using different techniques, by \cite{bacchiocchi2025} and \cite{lykouris2025contextual}. The former develops a general framework for online learning with piecewise-linear rewards, while the latter obtains the optimal support dependence for non-contextual pricing through variance-aware zooming. Our work extends these one-dimensional guarantees to the multidimensional setting arising in multilateral trade, where the number of relevant candidate decisions can grow as $O(K^d)$.

The key structural observation underlying these results is that finite supports induce only finitely many relevant actions. Between two consecutive valuation thresholds, changing an offer does not affect whether the corresponding agent accepts and, consequently, does not change the probability of trade. Within such a region, moving the offer in the appropriate direction strictly improves the intermediary's margin. Thus, if all support locations were known, it would suffice to optimize over the finite grid induced by the marginal supports. In our $d$-dimensional setting, this grid contains at most  $O(K^d)$ points, generalizing the $O(K)$ size in pricing. Running a standard finite-armed bandit algorithm over these candidates would then yield the desired $\widetilde O(\sqrt{K^dT})$ regret bound. The challenge is that the support locations, and hence the candidate set itself, are unknown. The main difficulty is therefore to discover sufficiently accurate candidate actions from bandit feedback while avoiding a prohibitively fine discretization of the entire action space.

\begin{figure*}[t]
  \centering
  \resizebox{\textwidth}{!}{
\begingroup
\definecolor{nzblue}{HTML}{245B8A}
\definecolor{nzgreen}{HTML}{23765A}
\definecolor{nzorange}{HTML}{B86A19}
\definecolor{nzgray}{HTML}{64717D}
\begin{tikzpicture}[
  x=1cm,y=1cm, font=\footnotesize,
  >={Latex[length=1.5mm,width=1mm]},
  axis/.style={->,thin,nzgray},
  cell/.style={thin,nzgray},
  guide/.style={densely dashed,nzgray!55,very thin},
  heading/.style={font=\bfseries\footnotesize,anchor=base},
  equation/.style={font=\footnotesize,align=center},
  note/.style={font=\scriptsize,align=center,text=nzgray},
  dot/.style={circle,fill=nzblue,inner sep=1.5pt}
]
\path[use as bounding box] (0,-1.8) rectangle (17.4,4.4);

\fill[nzblue!5] (2.6,1.15)--(5.15,.5)--(5.15,3.05)--(2.6,2.2)--cycle;
\draw[nzblue!60,densely dashed,thin] (2.6,2.2)--(5.15,3.05);
\draw[nzblue!60,densely dashed,thin] (2.6,1.15)--(5.15,.5);
\begin{scope}
  \node[heading] at (2.05,4.12) {(a) Support-induced grid};
  \node[equation] at (2.05,3.62) {$\Pr(V\preceq x)$ is constant\\within each cell};
  \fill[nzblue!10] (1.55,1.15) rectangle (2.6,2.2);
  \foreach \xx in {.65,1.55,2.6,3.65}
    \draw[guide] (\xx,.3)--(\xx,3.15);
  \foreach \yy in {.45,1.15,2.2,3.05}
    \draw[guide] (.45,\yy)--(3.75,\yy);
  \draw[nzblue,semithick] (1.55,1.15) rectangle (2.6,2.2);
  \foreach \xx in {.65,1.55,2.6,3.65}
    \foreach \yy in {.45,1.15,2.2,3.05}
      \fill[nzgray] (\xx,\yy) circle (.85pt);
  \draw[axis] (.35,.2)--(3.95,.2) node[above right,text=black,inner sep=1pt] {$x_1$};
  \draw[axis] (.35,.2)--(.35,3.25) node[above,text=black] {$x_2$};
  \foreach \pos/\lab in {1.55/a_1,2.6/b_1} {
    \draw[nzgray] (\pos,.2)--(\pos,.13);
    \node[below,font=\scriptsize,inner sep=2pt] at (\pos,.13) {$\lab$};
  }
  \foreach \pos/\lab in {1.15/a_2,2.2/b_2} {
    \draw[nzgray] (.35,\pos)--(.28,\pos);
    \node[left,font=\scriptsize,inner sep=2pt] at (.28,\pos) {$\lab$};
  }
  \node[dot] at (1.55,1.15) {};
  \node[below left,inner sep=3pt] at (1.55,1.15) {$a$};
  \node[equation] at (2.05,-1.28) {$\mu(a)\ge \mu(x)$ within the cell};
\end{scope}

\foreach \offset/\title/\prob/\radius/\col in {
  4.35/{(b) Large trade prob.}/{\Pr(V\preceq x)=1}/.72/nzgreen,
  8.7/{(c) Small trade prob.}/{0<\Pr(V\preceq x)<1}/1.8/nzorange,
  13.05/{(d) Normalized gap}/{\Pr(V\preceq x)>0}/.72/nzblue}
{
  \begin{scope}[xshift=\offset cm]
    \node[heading] at (2.05,4.12) {\title};
    \node[equation] at (2.05,3.62) {$\prob$};
    \fill[\col!16] (.8,.5)--({.8+\radius},.5)--(.8,{.5+\radius})--cycle;
    \draw[nzblue,semithick] (.8,.5) rectangle (3.65,3.05);
    \draw[\col,thick] ({.8+\radius},.5)--(.8,{.5+\radius});
    \draw[axis] (.35,.2)--(3.95,.2) node[above right,text=black,inner sep=1pt] {$x_1$};
    \draw[axis] (.35,.2)--(.35,3.25) node[above,text=black] {$x_2$};
    \foreach \pos/\lab in {.8/a_1,3.65/b_1} {
      \draw[nzgray] (\pos,.2)--(\pos,.13);
      \node[below,font=\scriptsize,inner sep=2pt] at (\pos,.13) {$\lab$};
    }
    \foreach \pos/\lab in {.5/a_2,3.05/b_2} {
      \draw[nzgray] (.35,\pos)--(.28,\pos);
      \node[left,font=\scriptsize,inner sep=2pt] at (.28,\pos) {$\lab$};
    }
    \node[dot,fill=\col] at (.8,.5) {};
    \node[above right,inner sep=3pt] at (.8,.5) {$a$};
  \end{scope}
}
\begin{scope}[xshift=4.35cm]
  \draw[<->,nzgreen] (.8,-.43)--(1.52,-.43)
    node[midway,below=2pt,text=black] {$\varepsilon$};
  \node[equation] at (2.05,-1.28) {$\mu(a)-\mu(x)\le\varepsilon$};
\end{scope}
\begin{scope}[xshift=8.7cm]
  \draw[<->,nzorange] (.8,-.43)--(2.6,-.43);
  \node[font=\scriptsize,anchor=north] at (1.7,-.49)
    {$\varepsilon/\Pr(V\preceq x)$};
  \node[equation] at (2.05,-1.28) {$\mu(a)-\mu(x)\le\varepsilon$};
\end{scope}
\begin{scope}[xshift=13.05cm]
  \draw[<->,nzblue] (.8,-.43)--(1.52,-.43)
    node[midway,below=2pt,text=black] {$\varepsilon$};
  \node[equation] at (2.05,-1.35)
    {$\displaystyle\frac{\mu(a)-\mu(x)}{\Pr(V\preceq x)}\le\varepsilon$};
\end{scope}
\end{tikzpicture}
\endgroup
}
  \caption {(a) The marginal supports partition the action space into cells within which $\Pr(V\preceq x)$ is constant. Panels (b)--(d) illustrate the $\varepsilon$-optimal region inside a fixed cell for different trade probabilities (assuming point $a$ is optimal within the cell). Panels (b) and (c) correspond to large and small trade probabilities, respectively: the $\ell_1$-radius of the $\varepsilon$-optimal region scales as $\varepsilon/\Pr(V\preceq x)$, up to truncation by the cell boundaries. Consequently, when the trade probability is small, a much larger portion of the cell appears nearly optimal. Panel (d) shows that normalizing by the trade probability restores an $\varepsilon$-optimal region of the same scale as in the large-probability case in Panel (b).}
  \label{fig:normalized-zooming}
\end{figure*}
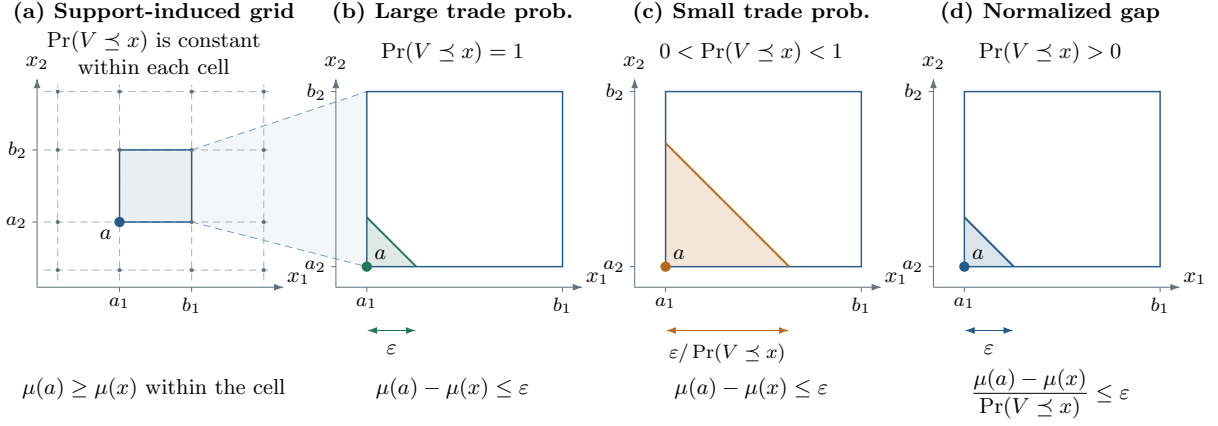

\subsection{Our Contributions}

We study the following abstract formulation, which captures the relevant structure of multilateral trade while simplifying the notation and slightly generalizing the model. At each round, the learner chooses an action $x\in[0,1]^d$ and receives reward
\[
\bigl(\beta-\langle \alpha,x\rangle\bigr)
\mathbf{1}\{V\preceq x\},
\]
where $V\in[0,1]^d$ is drawn from an unknown distribution, $\alpha\in(0,1]^d$, and $\beta$ is a constant such that rewards lie in $[-1,1]$. The learner observes only $\mathbf{1}\{V\preceq x\}$. Multilateral trade can be reduced to this formulation by reversing the coordinates corresponding to buyer prices: a buyer with value $v_i$ accepting a price $p_i$ is equivalently represented by $1-v_i\le 1-p_i$, whereas seller offers already have the required monotonicity.

Recall that $K_i$ denotes the support size of coordinate $i$ and let $K=\max_i K_i$. For any constant dimension $d$, we design an algorithm with regret
$
O\left(\sqrt{T K^d}\log T\right)
$.
The same guarantee therefore applies to multilateral trade with a constant number of agents.

\paragraph{A Small Set of Candidate Solutions}
Our starting point is a simple geometric property of finitely supported distributions. The support points of each coordinate partition $[0,1]^d$ into axis-aligned hyperrectangles. Within the interior of each such rectangle, the probability $\Pr(V\preceq x)$ is constant. Hence, within a rectangle, the expected reward is the product of a constant trade probability and the strictly decreasing linear function $\beta-\langle\alpha,x\rangle$. The best point in the rectangle is therefore its componentwise smallest corner.
As a consequence, when the supports are known, an optimal action can be found among the grid points induced by the marginal supports. There are at most
$O(K^d)$
such candidates. Running a standard finite-armed bandit algorithm on this set would immediately yield the desired support-dependent regret. The challenge is that the locations of these candidate points are unknown.

\paragraph{Why a Direct Zooming Approach Fails}
The observation above suggests learning the relevant grid adaptively rather than discretizing the entire hypercube in advance. A natural approach is to use a zooming-style algorithm that recursively refines only regions that may contain near-optimal actions.
Away from the support thresholds, the expected reward is
$
\mu(x)= \Pr(V\preceq x)\bigl(\beta-\langle\alpha,x\rangle\bigr),
$
and, since $\Pr(V\preceq x)$ is constant within each cell, its local slope with respect to coordinate $i$ is the constant $-\alpha_i \Pr(V\preceq x)$. Hence, the variation of the objective inside a cell is not determined by the cell diameter alone, but is also scaled by the probability of trade. When this probability is small, the objective can become nearly flat over a large region. A standard zooming rule may then regard the entire region as only mildly suboptimal and continue refining it unnecessarily, preventing a regret bound that depends only on the number of support-induced candidates.

\paragraph{Probability-normalized Zooming}
We overcome the problem by introducing a probability-normalized splitting rule. Rather than deciding whether to refine a cell solely from its geometric width, our algorithm scales the splitting rule using an estimate of the probability of trade in that cell. 
This normalization balances two competing requirements. Near an optimal support threshold, the algorithm must refine sufficiently aggressively to identify an accurate candidate action. Conversely, regions with very small trade probability contribute little variation in expected reward and should not be repeatedly subdivided. We show that this rule allows the adaptive partition to exploit the finite-support structure without knowing the support locations in advance, leading to the regret bound
$
O\left(\sqrt{T K^d}\log T\right)
$.
Figure~\ref{fig:normalized-zooming} depicts the core intuition underlying probability-normalized zooming.

\paragraph{A Tight Lower Bound}
We complement our upper bound with a matching support-dependent lower bound.
For every fixed dimension $d$, we show that any algorithm suffers regret
\[
\Omega \!\left(
\min\left\{
\sqrt{K^d T},
T^{2/3}
\right\}
\right)
\]
on some instance with marginal supports of size at most $K$.
Thus, our upper bound is optimal up to logarithmic factors whenever
$K \le T^{1/3d}$. Characterizing the optimal regret bound in the large-support regime remains an open problem. The best known upper bound is $\widetilde O (T^{3/4})$ in arbitrary dimension~\citep{castiglioni2026sample} and $\widetilde O(T^{7/10})$ in the two-dimensional case \citep{castiglioni2026breakingt34barrierregret}.
At a high level, our lower bound is based on a construction with $M^d$ disjoint cells, whose number grows exponentially with the number of agents. Under a
reference instance, the lower corner of each cell attains the same expected
reward. We then construct a family of alternative instances by locally
perturbing the distribution within a single cell, making its lower corner
optimal while leaving the feedback distribution unchanged outside that
cell. Hence, identifying the optimal action requires the learner to
distinguish among $M^d$ possible cells.
%, as in an $M^d$-armed bandit problem. 
Finally, a change-of-measure argument then yields the stated lower bound.
Finally, the construction is itself a multilateral-trade instance with
$d$ buyers and no sellers, after reversing the buyers' coordinates as in
our reduction.

\subsection{Related Work}
While the focus of this work is on profit maximization, there is a long line of research studying the related problem of gain from trade maximization. \citet{cesa2024bilateral} established that, under one-bit feedback, learning is impossible in general. They also showed that, under two-bit feedback, strongly budget-balanced mechanisms can be learned when buyer and seller values are independently drawn from sufficiently regular distributions.

Later works considered adversarial environments. \citet{azar2022alpha} obtained guarantees in terms of $2$-regret, whereas \citet{cesa2024regret} showed that standard no-regret guarantees can be recovered under a smoothed-adversary assumption. Relaxing the budget-balance requirement provides another route to meaningful guarantees without relying on smoothing. In particular, \citet{bernasconi2024glob} and \citet{chen2025tight} study how globally budget-balanced mechanisms can be exploited to obtain learning guarantees in fully adversarial environments. \citet{lunghiBetter} study instead the trade-off between regret and violation of the budget-balance constraint. These results, however, compare against the best fixed price in hindsight. For stronger randomized benchmarks, the picture is substantially different: no-regret learning is impossible in the adversarial setting \citep{bernasconi2024glob}.

Stronger guarantees become possible in stochastic environments. \citet{lunghi2026stronger} show that no-regret learning with respect to a distributional benchmark can be achieved under suitable smoothness assumptions.

Beyond the classical bilateral-trade model, the literature has considered several extensions, including contextual environments \citep{gaucher2025featurebased,cosson2026contextual,coccianonparametric}, fairness considerations \citep{bachoc2024fair}, markets with multiple buyers \citep{babaioff2024learning,lunghi2025online}, and related models in which the two sides of the market play more
symmetric roles
\citep{bacchiocchiOnline2025,bolic2024brokerage,bachoc2025b};
see also \citep{cesari2025an,bachoc2025a}.

\section{Preliminaries}

At round $t$, a learner selects $x_t\in[0,1]^d$. The valuations $V_t$ in each round $t$ are independently sampled from an unknown joint distribution $\mathcal{D}$ supported on $[0,1]^d$, and the learner observes only
\[
Z_t= \1\{V_t\preceq x_t\},
\]
where we write $x\preceq y$ if $x_i\le y_i$ for every $i$. 
The learner can always compute also the  reward
\[
Y_t=(\beta-\langle \alpha,x_t\rangle) \1\{V_t\preceq x_t\},
\]
where we assume that the scalar $\beta$ and the vector $\alpha\in (0,1]^d$ are known to the learner, and satisfy
\[(\beta-\langle \alpha,x\rangle) \1\{V_t\preceq x\}\in [-1,1] \quad \forall x\in [0,1]^d.\]
To shorten the notation we set $g(x)=(\beta-\langle \alpha,x\rangle)$ and we assume $g(x)\in [-1,1]$ for each $x\in [0,1]^d$.

We define the lower-orthant CDF and mean reward by
\[
 F(x)=\Pr(V\preceq x),\qquad \mu(x)= (\beta- \langle \alpha, x \rangle) F(x).
\]
Letting  $\mu^\star=\sup_x\mu(x)$ and $\Delta(x)=\mu^\star-\mu(x)$, we define the pseudo-regret as
\[
 \Reg_T=\E\left[\sum_{t=1}^T\Delta(x_t)\right].
\]

Let $S_i$ be the projection of $\supp(\cD)$ onto coordinate $i$, and let
\(
 K_i=|S_i|,\, K=\max_i K_i<\infty,\,
 G=\prod_{i=1}^d S_i.
\)
Then $|G|\le K^d$. Notice that the distribution need not have independent coordinates.

For every $x\in[0,1]^d$ such that $x_i\ge \min_{z\in S_i}z \, \forall i\in[d]$, define its coordinatewise rounding $y(x)\in G$ by
\[
 y_i(x)=\max\{z\in S_i:z\le x_i\}.
\]
Notice that $F(x)=F(y(x))$.
Thus $\mu(y(x))\ge\mu(x)$, and a maximizer of $\mu$ exists in the finite grid $G\cup \{ \textbf{0} \}$.\footnote{Adding the origin $\mathbf{0}$ to the grid addresses the case in which $g(x)$ is negative at every point in $G$.}

For the rest of this paper we will assume that there exists an $x\in [0,1]^d$ such that $g(x)\ge0$, otherwise the problem is trivial as an optimal action always lies at $\mathbf0$.

\paragraph{A Special Case: Multilateral Trade}

Consider an intermediary facilitating a transaction between a set
of buyers $B$ and a set of sellers $S$. Write $m=|B|$, $n=|S|$,
and $d=m+n$. At each round $t$, the intermediary posts a price
$p_{t,i}\in[0,1]$ to each buyer $i$ and offers a payment
$q_{t,j}\in[0,1]$ to each seller $j$. The valuation vector
$(v_t^b,v_t^s)$ is drawn independently across rounds from an
unknown joint distribution, allowing arbitrary dependence
across agents.
Multilateral trade admits a CDF representation by reversing
the buyers' coordinates. Set $x_i=1-p_i$ and $V_i=1-v_i^b$
for each buyer, and $x_{|B|+j}=q_j$ and
$V_{|B|+j}=v_j^s$ for each seller. Trade then occurs exactly
when $V\preceq x$, and the expected profit, normalized in the interval $[-1,1]$, is
\[
\mu(x)=g_{\mathrm{tr}}(x)F(x),
\qquad
g_{\mathrm{tr}}(x)=\frac{m-\sum_{i=1}^{d}x_i}{d}.
\]

\section{Why Does Standard Zooming Fail?}
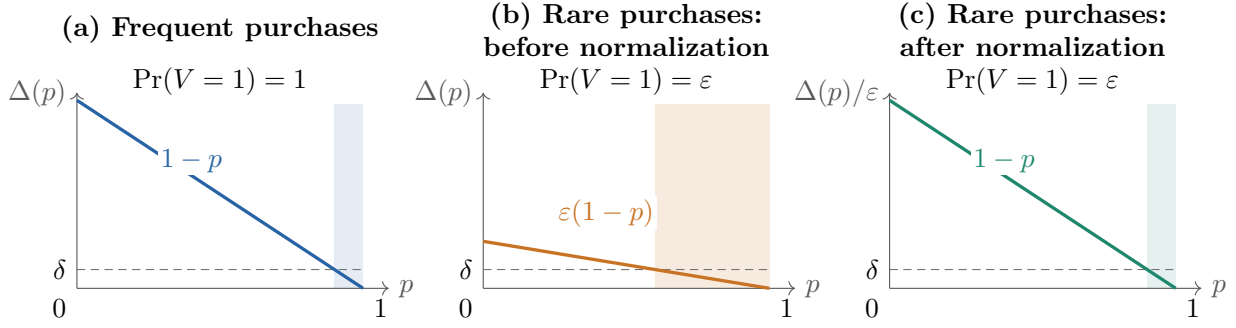
\begin{figure*}[t]
  \centering
  \resizebox{\textwidth}{!}{%
    % Requires: \usepackage{tikz}
% Illustration uses r=1/10 and epsilon=1/4; labels state general counts.
\begin{tikzpicture}[
x=3.65cm,y=2.4cm,
=stealth,
every node/.style={font=\small},
axis/.style={->,black!65,thin},
curve/.style={very thick},
threshold/.style={densely dashed,black!55,thin},
cell/.style={thin,black!55},
paneltitle/.style={font=\small\bfseries,align=center},
count/.style={font=\small,align=center}
]
\definecolor{zoomblue}{RGB}{40,100,165}
\definecolor{zoomorange}{RGB}{200,115,35}
\definecolor{zoomgreen}{RGB}{30,135,110}
% Panel 1: purchase probability one (representative constant probability).
\begin{scope}
\node[paneltitle] at (.5,1.38) {(a) Frequent purchases};
\node at (.5,1.10) {$\Pr(V=1)=1$};
\fill[zoomblue!12] (.9,0) rectangle (1,.98);
\draw[axis] (0,0)--(1.09,0) node[right] {$p$};
\draw[axis] (0,0)--(0,1.04) node[ left] {$\Delta(p)$};
\draw[threshold] (0,.1)--(1,.1);
\node[left] at (0,.1) {$\delta$};
\draw[curve,zoomblue] (0,1)--(1,0);
\node[zoomblue,fill=white,inner sep=1pt] at (.40,.68) {$1-p$};
\node[below left] at (0,0) {$0$};
\node[below right] at (1,0) {$1$};

\end{scope}
% Panel 2: the same absolute gap threshold covers a wider price interval.
\begin{scope}[shift={(1.42,0)}]
\node[paneltitle] at (.5,1.38) {(b) Rare purchases:\\
before normalization};
\node at (.5,1.10) {$\Pr(V=1)=\varepsilon$};
\fill[zoomorange!14] (.6,0) rectangle (1,.98);
\draw[axis] (0,0)--(1.09,0) node[right] {$p$};
\draw[axis] (0,0)--(0,1.04) node[ left] {$\Delta(p)$};
\draw[threshold] (0,.1)--(1,.1);
\node[left] at (0,.1) {$\delta$};
\draw[curve,zoomorange] (0,.25)--(1,0);
\node[zoomorange,fill=white,inner sep=1pt] at (.43,.40) {$\varepsilon(1-p)$};
\node[below left] at (0,0) {$0$};
\node[below right] at (1,0) {$1$};
\end{scope}
% Panel 3: normalization represents the probability-weighted gap threshold.
\begin{scope}[shift={(2.84,0)}]
\node[paneltitle] at (.5,1.38) {(c) Rare purchases: \\
after normalization};
\node at (.5,1.10) {$\Pr(V=1)=\varepsilon$};
\fill[zoomgreen!12] (.9,0) rectangle (1,.98);
\draw[axis] (0,0)--(1.09,0) node[right] {$p$};
\draw[axis] (0,0)--(0,1.04) node[ left] {$\Delta(p)/\varepsilon$};
\draw[threshold] (0,.1)--(1,.1);
\node[left] at (0,.1) {$\delta$};
\draw[curve,zoomgreen] (0,1)--(1,0);
\node[zoomgreen,fill=white,inner sep=1pt] at (.40,.68) {$1-p$};
\node[below left] at (0,0) {$0$};
\node[below right] at (1,0) {$1$};
\end{scope}
\end{tikzpicture}%
  }
  \caption{Panels (a) and (b) indicate with shaded intervals the sets of $\delta$-approximately optimal prices for different values of $\Pr(V=1)$, assuming that $\Pr(V=0)=1-\Pr(V=1)$. Panel (c) illustrates how normalization makes the slope independent of the purchase probability.}
  \label{fig:weighted-zooming}
\end{figure*}
To better explain why the regret bound for the classical zooming
algorithm~\citep{kleinberg2008multiarmed} can be larger than the
one obtained by our algorithm, we start with an introductory example.
Consider a setting with two buyer types. The first has valuation
zero, and the second has valuation one.

If the buyer with valuation one appears with probability $\Theta(1)$,
a purchase occurs with constant probability at every positive price.
For every $\delta>0$, the $\delta$-optimal prices can then
be covered by a constant number of intervals of diameter $\delta$.
Therefore, the algorithm by~\cite{kleinberg2008multiarmed} keeps refining intervals concentrated near the
optimal price, yielding a regret bound of $\widetilde O(\sqrt T)$.

When the buyer with valuation one appears with probability
$\varepsilon$, the expected revenue
at price $p$ is $\varepsilon p$, and the suboptimality gap is
$\varepsilon(1-p)$. The number of intervals of diameter $\delta$ needed
to cover the $\delta$-optimal prices is \(\Theta\!\left(\min\left\{1/\delta,1/\varepsilon\right\}\right).\) 
Therefore, the total number of intervals that can be
refined throughout the execution is bounded by
$\widetilde O(1 /\varepsilon)$. Following the analysis
of~\citet{kleinberg2008multiarmed}, this yields a regret bound of
$\widetilde O\left(
\min\left\{\sqrt{{T}/{\varepsilon}},\,\varepsilon T\right\}
\right)$.
One of the main contributions of our paper is to address this problem through a normalization that increases the slope of the reward function.
Figure \ref{fig:weighted-zooming} illustrates an example of why zooming fails and how normalization fixes this issue.

\section{Suboptimality and the support grid}

To obtain zooming-style guarantees that adapt to the support of the distribution, and in particular to the cardinality of \(G\), we first need to characterize the suboptimality gap of points that do not lie in a sufficiently small neighborhood of \(G\). The general idea is that, by exploiting the structure of the $\mu(\cdot)$ function, one can show that all points at most $\varepsilon$ suboptimal lie in a neighborhood of a point of the marginal support $G$.

To do so, two considerations are in order. 
First, the problem is inherently asymmetric: the objective
\[
\mu(x)=g(x)F(x)
\]
is, in general, only right-continuous, which means that any notion of neighborhood should reflect the asymmetry. Second, while $g(\cdot)$ is  globally Lipschitz, it behaves very differently along different directions, so the notion of neighborhood should vary across directions.

To account for both issues, we introduce an asymmetric notion of neighborhood defined through a norm adapted to the function $g$. For \(y\in G\) and \(\varepsilon>0\), define the upper ball of radius \(\varepsilon\) centered at \(y\) by
\[
B^{\uparrow,\alpha}_\varepsilon(y)
\coloneqq
\left\{x\succeq y:\|x-y\|_{1,\alpha}\le \varepsilon\right\},
\]
where the norm $\|\cdot\|_{1,\alpha}$ is defined as the weighted norm
\[\|z\|_{1,\alpha}=\sum_{i=1}^d\alpha_i|z_i|.\]
Thus, only points lying above \(y\) in the coordinatewise order are regarded as neighboring \(y\), and therefore  $x\in B^{\uparrow,\alpha}_\varepsilon(y)$ means that $g(y)-g(x)\le \varepsilon$.

Finally, we have to account for the fact that $\mu(\cdot)$ is a product of $g(\cdot)$ and the CDF $F(\cdot)$, that makes the local slope of $\mu(\cdot)$ proportional to the CDF in the region. Therefore, we can only derive a lower bound on the suboptimality gap of a point \(x\) that scales linearly with \(F(x)\). This weighted bound is sufficient for our purposes and captures the local regularity of the objective away from the asymmetric neighborhoods of the support grid.
%
%Set $B^\uparrow_\varepsilon(y)=\{x\succeq y:\|x-y\|_1\le\varepsilon\}$.
%
Formally:
\begin{restatable}{lemma}{LemmaPalla}\label{lem:gap}
For any $\epsilon>0$ and any $x\in [0,1]^d$, if $x\notin\bigcup_{y\in G}B^{\uparrow,\alpha}_\varepsilon(y)$, then
$\Delta(x)\ge \varepsilon F(x)$.
Additionally, if $F(x)>0$ it also holds 
$\Delta(x)> \varepsilon F(x).$
\end{restatable}

\section{Algorithm Description}

At a high level, Algorithm~\ref{alg:zoom} follows the main principle of zooming algorithms: it starts from a coarse discretization of the action space and progressively refines the discretization only in regions that remain potentially optimal. The algorithm implements this idea through an optimistic selection rule, together with a refinement criterion that balances the error due to statistical uncertainty against the error induced by the current spatial approximation. 

More specifically, Algorithm~\ref{alg:zoom} maintains an adaptive partition $\mathcal P$ of $[0,1]^d$ into axis-aligned rectangles. For a cell $A$, let $\ell_A,u_A$ be the lower and upper corners of its closure, and define
\begin{align}
 \omega_A&=\sum_{i=1}^d \alpha_i(u_{A,i}-\ell_{A,i})
 =g(\ell_A)-g(u_A).\label{eq:gwidth}
\end{align}

The quantity $\omega_A$ is the exact variation of $g$ on the cell. Given an upper confidence bound $F_A^+$ for $F(u_A)$, the algorithm selects a cell maximizing $U_A=g(\ell_A)F_A^+$. At a high level, the algorithm refines the selected cell if its statistical uncertainty is smaller than its probability-weighted geometric approximation error:
\begin{equation}\label{eq:split-rule}
 g(\ell_A)\beta_A<\omega_AF_A^+,
\end{equation}
provided the maximum depth has not been reached. To avoid over-splitting, when the optimum is small, and therefore cells with very low $g(\ell_A)$ are selected, we add an additional condition
\(4\beta_A<F_A^+,\)
which is an active condition only when it holds $g(\ell_A)<4(g(\ell_A)-g(u_A))$.
Otherwise it plays $u_A$ and updates that cell's estimate.

\paragraph{Bisection of a Largest Weighted Side}
To split $A$, choose
\[
 i_A=\arg\max_{i\in[d]} \left(\alpha_i(u_{A,i}-\ell_{A,i})\right)
\]
and bisect only coordinate $i_A$ at its geometric midpoint. This creates two children and halves the largest weighted side. Any ties are broken deterministically, ensuring the entire binary tree remains deterministic even though only an adaptive subset of its nodes is generated.
At depth $h$, all cells of the full tree have the same side-length vector. Let $k_i(h)$ be the number of bisections in coordinate $i$ along a depth-$h$ branch. These counts are independent of the branch, since the choice of split coordinate depends only on the side lengths (see Figure \ref{fig:adaptive-cell-depth}). 
After at most $d$ further bisections every weighted side is at most half the previous maximum: each side larger than that threshold is halved before any smaller side is selected. Since $\omega_0\le 2$, where $\omega_0$ is $g$-diameter of the initial cell $[0,1]^d$, we have
\begin{equation}\label{eq:depth-error}
 \omega_h\le \,2^{-\lfloor h/d\rfloor+1}.
\end{equation}
This rule avoids excessive refinement along coordinates with small coefficients, making our method suitable even for very small choices of $\alpha_i$.

% Preamble:
% \usepackage{tikz}
% \usetikzlibrary{arrows.meta,calc}

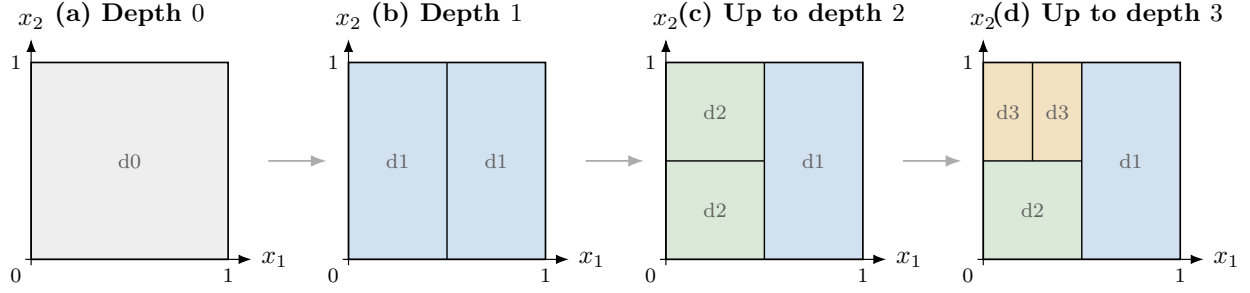
\begin{figure*}[t]
\centering
\resizebox{\textwidth}{!}{%
\begin{tikzpicture}[
    >=Latex,
    font=\small,
    boundary/.style={black, line width=0.65pt},
    cut/.style={black, line width=0.55pt},
    axis/.style={->, line width=0.45pt},
    transition/.style={
        ->,
        >=Latex,
        gray!65,
        line width=0.7pt
    },
    paneltitle/.style={font=\small\bfseries},
    depthlabel/.style={
        font=\scriptsize,
        text=black!62
    },
]

% ------------------------------------------------
% Muted palette
% ------------------------------------------------
\definecolor{rootgray}{RGB}{238,238,238}
\definecolor{depthone}{RGB}{210,225,239}
\definecolor{depthtwo}{RGB}{218,233,216}
\definecolor{depththree}{RGB}{243,226,194}

% Geometry
\def\S{2.7}
\def\H{1.35}
\def\Q{0.675}
\def\sep{4.35}

% =================================================
% (a) DEPTH 0
% =================================================
\begin{scope}[shift={(0,0)}]

    \node[paneltitle] at (1.35,3.35)
        {(a) Depth $0$};

    \fill[rootgray]
        (0,0) rectangle (\S,\S);

    \draw[boundary]
        (0,0) rectangle (\S,\S);

    % axes
    \draw[axis]
        (-0.08,0) -- (3.02,0)
        node[right] {$x_1$};

    \draw[axis]
        (0,-0.08) -- (0,3.02)
        node[above] {$x_2$};

    \node[below left, font=\scriptsize]
        at (0,0) {$0$};

    \node[below, font=\scriptsize]
        at (\S,0) {$1$};

    \node[left, font=\scriptsize]
        at (0,\S) {$1$};

    % depth label
    \node[depthlabel]
        at (1.35,1.35)
        {d$0$};

\end{scope}

% -------------------------------------------------
% transition
% -------------------------------------------------
\draw[transition]
    (3.25,1.35) -- (4.05,1.35)
    node[midway, above, font=\scriptsize, text=black!60]
    {};

% =================================================
% (b) DEPTH 1
% =================================================
\begin{scope}[shift={(\sep,0)}]

    \node[paneltitle] at (1.35,3.35)
        {(b) Depth $1$};

    % two children of the root
    \fill[depthone]
        (0,0) rectangle (\H,\S);

    \fill[depthone]
        (\H,0) rectangle (\S,\S);

    \draw[boundary]
        (0,0) rectangle (\S,\S);

    \draw[cut]
        (\H,0) -- (\H,\S);

    % axes
    \draw[axis]
        (-0.08,0) -- (3.02,0)
        node[right] {$x_1$};

    \draw[axis]
        (0,-0.08) -- (0,3.02)
        node[above] {$x_2$};

    \node[below left, font=\scriptsize]
        at (0,0) {$0$};

    \node[below, font=\scriptsize]
        at (\S,0) {$1$};

    \node[left, font=\scriptsize]
        at (0,\S) {$1$};

    % depth labels
    \node[depthlabel]
        at (0.675,1.35)
        {d$1$};

    \node[depthlabel]
        at (2.025,1.35)
        {d$1$};

\end{scope}

% -------------------------------------------------
% transition
% -------------------------------------------------
\draw[transition]
    (7.60,1.35) -- (8.40,1.35)
    node[midway, above, font=\scriptsize, text=black!60]
    {};

% =================================================
% (c) CELLS UP TO DEPTH 2
% =================================================
\begin{scope}[shift={({2*\sep},0)}]

    \node[paneltitle] at (1.35,3.35)
        {\hspace{0.6cm} (c) Up to depth $2$};

    % unsplit depth-1 cell
    \fill[depthone]
        (\H,0) rectangle (\S,\S);

    % children created at depth 2
    \fill[depthtwo]
        (0,0) rectangle (\H,\H);

    \fill[depthtwo]
        (0,\H) rectangle (\H,\S);

    % boundary
    \draw[boundary]
        (0,0) rectangle (\S,\S);

    % depth-1 cut
    \draw[cut]
        (\H,0) -- (\H,\S);

    % only the selected depth-1 cell is refined
    \draw[cut]
        (0,\H) -- (\H,\H);

    % axes
    \draw[axis]
        (-0.08,0) -- (3.02,0)
        node[right] {$x_1$};

    \draw[axis]
        (0,-0.08) -- (0,3.02)
        node[above] {$x_2$};

    \node[below left, font=\scriptsize]
        at (0,0) {$0$};

    \node[below, font=\scriptsize]
        at (\S,0) {$1$};

    \node[left, font=\scriptsize]
        at (0,\S) {$1$};

    % depth labels
    \node[depthlabel]
        at (0.675,0.675)
        {d$2$};

    \node[depthlabel]
        at (0.675,2.025)
        {d$2$};

    \node[depthlabel]
        at (2.025,1.35)
        {d$1$};

\end{scope}

% -------------------------------------------------
% transition
% -------------------------------------------------
\draw[transition]
    (11.95,1.35) -- (12.75,1.35)
    node[midway, above, font=\scriptsize, text=black!60]
    {};

% =================================================
% (d) CELLS UP TO DEPTH 3
% =================================================
\begin{scope}[shift={({3*\sep},0)}]

    \node[paneltitle] at (1.35,3.35)
        {\;\hspace{0.6cm}(d) Up to depth $3$};

    % ------------------------------------------------
    % cells inherited from previous depths
    % ------------------------------------------------

    % depth-1 cell: never refined
    \fill[depthone]
        (\H,0) rectangle (\S,\S);

    % lower-left depth-2 cell: not refined further
    \fill[depthtwo]
        (0,0) rectangle (\H,\H);

    % ------------------------------------------------
    % selected depth-2 cell is split into depth-3 cells
    % ------------------------------------------------
    \fill[depththree]
        (0,\H) rectangle (\Q,\S);

    \fill[depththree]
        (\Q,\H) rectangle (\H,\S);

    % boundary
    \draw[boundary]
        (0,0) rectangle (\S,\S);

    % inherited depth-1 split
    \draw[cut]
        (\H,0) -- (\H,\S);

    % inherited depth-2 split
    \draw[cut]
        (0,\H) -- (\H,\H);

    % newly created depth-3 split
    \draw[cut]
        (\Q,\H) -- (\Q,\S);

    % axes
    \draw[axis]
        (-0.08,0) -- (3.02,0)
        node[right] {$x_1$};

    \draw[axis]
        (0,-0.08) -- (0,3.02)
        node[above] {$x_2$};

    \node[below left, font=\scriptsize]
        at (0,0) {$0$};

    \node[below, font=\scriptsize]
        at (\S,0) {$1$};

    \node[left, font=\scriptsize]
        at (0,\S) {$1$};

    % depth labels
    \node[depthlabel]
        at (0.675,0.675)
        {d$2$};

    \node[
    depthlabel,
    text width=0.55cm,
    align=center
]
at (0.3375,2.025)
{d$3$};

    \node[
    depthlabel,
    text width=0.5cm,
    align=center
]
        at (1.0125,2.025)
        {d$3$};

    \node[depthlabel]
        at (2.025,1.35)
        {d$1$};

\end{scope}

\end{tikzpicture}%
}

\caption{
Illustration of the adaptive cell tree in two dimensions.
The root cell has depth $0$, and splitting a cell at depth $h$
creates two children at depth $h+1$. We use the notation d$i$ for cells of depth $i$. 
Only selected cells are refined, so the active partition may contain
cells at different depths.
Cells belonging to the same depth nevertheless have the same
side-length vector.
}
\label{fig:adaptive-cell-depth}
\end{figure*}

\paragraph{Depth and Confidence Parameters}
For a fixed horizon $T\ge1$ and confidence level $0<\delta<1$, set
\begin{equation}\label{eq: H c def}
 \begin{split}
 H=d\lceil\log_2(&dT)\rceil,\qquad N_H=2^{H+1}-1,\\
c&=2\log\frac{2N_HT}{\delta}.
 \end{split}
\end{equation}
The full binary tree has $N_H$ nodes up to depth $H$, and~\eqref{eq:depth-error} gives $\omega_H\le T^{-1}$. We use $c$ for the confidence parameter.
For a cell $A$, let $n_A$ count the plays assigned to it and let $\widehat F_A$ be the mean of its observations $Z_t$.
Initially, $n_A=0$ and $\widehat F_A=0$. We define the confidence radius
\begin{equation}\label{eq: def beta a}
 \beta_A=\begin{cases}1,&n_A=0,\\
 \sqrt{c/n_A},&n_A\ge1,
 \end{cases}
\end{equation}
and the upper confidence bound on the trade probability at $u_A$,
\begin{equation}\label{eq: def F +}
F_A^+:=\min\{1,\widehat F_A+\beta_A\},
\end{equation}
where $\beta_A$ controls the uncertainty due to the finite number of samples collected at the cell's upper corner.

\begin{algorithm}[t]
\caption{Probability-normalized zooming with upper-corner sampling}\label{alg:zoom}
\KwIn{$T,\delta$; set $H,c$,  as per \eqref{eq: H c def}%,\eqref{eq: def beta a},\eqref{eq: def F +}
}
Initialize $\mathcal P=\{[0,1]^d\}$ and its statistics\;
\For{$t=1,\ldots,T$}{
 \Repeat{an action has been played}{
  Compute $F_A^+$ and $U_A=g(\ell_A)F_A^+$ for every cell $A\in\mathcal P$\;
  Select $A\in\arg\max_{B\in\mathcal P}U_B$\;
  %Set $r_A=\|u_A-\ell_A\|_1$\;
  \eIf{$\operatorname{depth}(A)<H$ and %$g(\ell_A)\beta_A<(g(\ell_A)-g(u_A))F_A^+$}{
  $F_A^+> \max\left\{4, \frac{g(\ell_A)}{g(\ell_A)-g(u_A)}\right\}\beta_A$}{
  Set $i_A \in \arg\max_{i \in [d]}\alpha_i(u_{A,i}-\ell_{A,i})$\;
   Bisect coordinate $i_A$ and replace $A$ by its two children, initializing their statistics\;
  }{
   Play $x_t=u_A$ and observe $Z_t$\;
    Set $Y_t=g(u_A)\cdot Z_t$\;
    Set $\widehat F_A\gets(n_A\widehat F_A+Z_t)/(n_A+1)$ and then $n_A\gets n_A+1$\;
  }
 }
}
\end{algorithm}

\section{Algorithm Guarantees}

The proof works as follows. First, we show that the estimates of the probability are close to the true values. Then, we upper bound the number of cell splits. Finally, we combine the regret suffered in each cell to derive the final bound.

\paragraph{Clean Event}
Let $\mathcal E$ be the event that, for every created cell at every evaluation,
\begin{equation}\label{eq:confidence}
 |\widehat F_A-F(u_A)|\le\beta_A.
\end{equation}
We show that such an event holds with high probability. 
To do so, we apply Hoeffding' s inequality and a union bound over all possible cells and sample counts. This bounds with high probability the difference between the empirical CDFs and the its true value. 

\begin{restatable}{lemma}{LemmaClean}
    Event $\mathcal E$ holds with probability at least $1-\delta$.
\end{restatable}

Notice that under event $\mathcal{E}$, for every active cell $A$ and any round,
\begin{subequations}\label{eq:optimism}
\begin{align}
 &F(u_A)\le F_A^+\le F(u_A)+2\beta_A\\
 & \mu(x)\le U_A  \qquad\forall x\in A, g(\ell_A)\ge 0.
\end{align}
\end{subequations}

The latter follows from $0\le g(x)\le g(\ell_A)$ and $F(x)\le F(u_A)$. Additionally, there  always exists a cell $A'$ such that $g(\ell_{A'})\ge 0$, which implies that $U_{A'}\ge 0$ and that $g(\ell_A)\ge 0$ for each selected cell.
Therefore every selected cell satisfies $U_A\ge\mu^\star$.

\paragraph{Upperbounding the Number of Cells}

We begin by deriving a sufficient condition for a cell not to be split. In particular, we prove that under the event $\mathcal{E}$ a cell is split only if it is in the neighbor of an element in $G$, using the notion of $B^{\uparrow,\alpha}$ distance balls. The radius of such distance ball is defined by the $g$-radius of the cell itself ($\omega_A$), the smaller the cell, the closer $u_A$ must be to an element in $G$ to be split.

\begin{restatable}{lemma}{lemmaExclude}\label{lem:exclude}
Assume event $\mathcal E$ holds. Let $A$ be a cell such that $g(\ell_A)\ge 0$ and 
\[u_A\notin  B^{\uparrow,\alpha}_{6\omega_A}(y) \quad \forall y \in G\]
where $\omega_A=\|u_A-\ell_A\|_{1,\alpha}$. Then $A$ is never split by Algorithm \ref{alg:zoom}.
\end{restatable}

Given that only cells sufficiently close to points in $G$ can be split, we can easily derive an upper bound on the number of cells active at each depth.

\begin{restatable}{lemma}{lemmaNumber}\label{lem:split-count}
Under event $\mathcal E$, at every depth $h\le H$ at most
\(
 %O\left((K^d+1)2^d\left(\left\lceil\frac{6Ld}{\lambda}\right\rceil+2\right)^d\right)
 2\left(12d+2\right)^d|G|
\)
cells are ever created.
\end{restatable}

\paragraph{Regret Guarantee}
We combine the previous lemmas by bounding the regret accumulated over all rounds in which the same cell is selected. If a cell is still coarse, so that its geometric approximation error is large, it cannot remain active for long: once the statistical uncertainty becomes sufficiently small, the splitting condition is triggered and the cell is replaced by its two children. Since Lemma \ref{lem:split-count} provides a bound on the total number of active cells, we can then aggregate the regret contributions across cells and apply the Cauchy-Schwarz inequality to obtain the following regret bound:

\begin{theorem}\label{thm:regret}
 Algorithm~\ref{alg:zoom} with $\delta=1/T$ satisfies
\[
 \Reg_T =O\!\left(\sqrt{K^dT}\log T\right).
\]
\end{theorem}
\begin{proof}
Condition on $\mathcal E$. By Lemma \ref{lem:split-count}, summing over the depths and including the root cell, we get that at most 
\(B= 1+2H(12d+2)^d|G|\)
cells are ever created.

For a cell $A$, let $m_A$ be the number of times cell $A$ is played. Consider the round in which it is played  for the $m_A$-th times. The same calculation as in \eqref{eq:split-gap} yields

\[\Delta(u_A)\le 2 \beta_A g(\ell_A) + \omega_A F(u_A)\] 

We consider two cases. If level smaller than $H$,
since it has not been split, we get that \( \omega_A F(u_A)\le  \beta_A g(\ell_A), \) 
or 
\( F(u_A)\le  4\beta_A \)
which implies \[\Delta(u_A)\le 2 \beta_A g(\ell_A) + \omega_A F(u_A)\le 10 \beta_A,\]
using $\omega_A\le 2$ and $g(\ell_A)\le 1$.
%Otherwise, we get 
%\[ d 2^{-H} L F(u_A)\le  \beta_A g(\ell_A) \]\textcolor{orange}{questo secondo caso non mi torna molto}
Otherwise, if the depth is $H$ we have that $$\omega_H\le d2^{-\lfloor H/d \rfloor}\le 1/T $$ and 
\[\Delta(u_A)\le 10 \beta_A  + T^{-1} F(u_A).\]

Additionally, let $\beta_{A,j}$ the value of $\beta_A$ in the rounds where $n_A=j-1$, we get:
\begin{align*}
\sum_{j=1}^{m_A}\beta_{A,j}
 &\le1+\sqrt c\sum_{n=1}^{m_A-1}n^{-1/2}\\
 &\le1+2\sqrt{cm_A}\le(1+2\sqrt c)\sqrt{m_A}
\end{align*}
and by applying the Cauchy-Schwarz inequality 
\begin{align*}
    \sum_{t=1}^T\Delta(x_t)&\le \sum_{t=1}^T \left(10\beta_{A(x_t),j_t}+T^{-1} F(u_{A(x_t)})\right)\\
    & \le 10(1+2\sqrt c)\sum_{A:m_A>0}\sqrt{m_A}+ 2T\,T^{-1}\\
    & \le 10(1+2\sqrt c)\sqrt{BT}+2,
\end{align*}
where $A(x_t)$ is the active cell at time $t$ such that $x_t=u_A$ and $j_t$ is the number of rounds cell $A(x_t)$ has been played till round $t-1$.
When $\mathcal E$ does not hold, the cumulative gap is at most $2T$. Taking expectations adds at most $\delta 2T$. Finally, $\delta=1/T$ concludes the proof.
\end{proof}

%For the camera-ready paper, if you are using \LaTeX, please make sure
%that you follow these instructions.
% (If you are not using \LaTeX,
%please make sure to achieve the same effect using your chosen
%typesetting package.)

\begin{figure*}[t]
	\centering
	\includegraphics[width=0.75\textwidth]{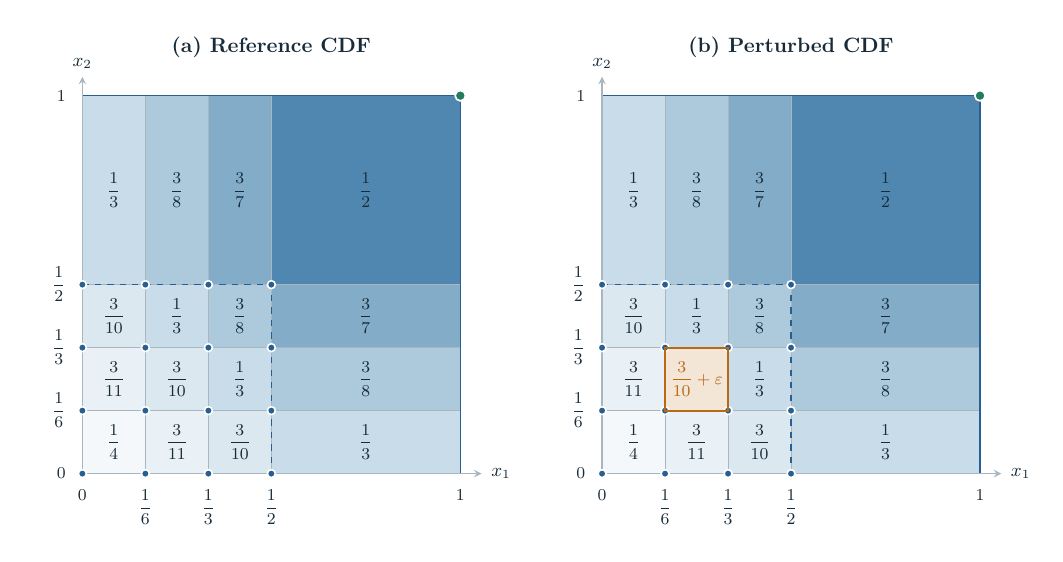}
    \vspace{-0.7cm}
\caption{Reference (a) and perturbed (b) CDFs for $d=2$ and $M=3$.
	Both distributions have support $ \mathcal{Z} \cup\{(1,1)\}$, with $\mathcal{Z} = \{0,1/6,1/3,1/2\}^2 $ .
	Equal blue shades indicate equal CDF values; the values written in the
	rectangles are exact. In panel (b), the perturbed cell is $A_j$, with $j=(1,1)$.}
	\label{fig:lower-bound-instances}
\end{figure*}

\section{A Tight Lower Bound}
\label{sec:lower-bound}
We establish a support-dependent regret lower bound, showing that the bound in Theorem~\ref{thm:regret} is tight up to logarithmic factors whenever $K\le {T^{1/(3d)}}$.  Characterizing the optimal regret bound in the large-support regime remains an open problem. The best known upper bound is $\tilde O (T^{3/4})$ in arbitrary dimension~\cite{castiglioni2026sample} and $\tilde O(T^{7/10})$ in the two-dimensional case \cite{castiglioni2026breakingt34barrierregret}.

\begin{theorem}\label{thm:lower-bound}
For every integer $d\ge1$, there exists a constant $\kappa_d>0$ such that, for all integers $K\ge4$ and $T\ge1$ and every algorithm, there exists an instance with marginal supports of size at most $K$ for which 
\[
\Reg_T \ge \kappa_d \min\left\{\sqrt{K^dT},T^{2/3}\right\}.
\]
\end{theorem}

We remark that the lower bound holds even for multilateral trade models with $d$ buyers and no sellers.

While the formal proof is deferred to Appendix \ref{app:lower-bound-proof}, we provide here a sketch of the construction.
We construct a family of instances with a common support, consisting of an equally spaced grid and the additional point $\mathbf{1}$. The instances differ only in the probabilities assigned to these points.

Fix $M\ge2$ and let $\rho=1/(2M)$. Each agent’s marginal support is $S_i=\{0,\rho,\ldots,M\rho,1\}$, while the joint support is $ \mathcal{Z} \cup\{\mathbf{1}\}$, where $\mathcal{Z}:=\{0,\rho,\ldots,M\rho\}^d$. This grid defines $M^d$ \emph{disjoint} cells in $[0,1/2]^d$, whose number grows exponentially with the number of agents. In addition, in all instances we take  \(g(x)=1-\sum_{i} x_i/d.\) %\mat{Non ideale che non è un istanza di multilateral trade. Credo non ci sia modo di fixarla formalmente dicendo che stiamo applicando una traslazione. D'altra parte è ovvio come estendere la proof.}

We index the grid points belonging to $\mathcal{Z}$ by $j\in\{0,\ldots,M\}^d$ and write $z_j=\rho j$.
We start from a reference distribution over valuations whose CDF satisfies
\[
F_0(z)=\frac{1}{4g(z)}
\qquad \text{for every } z\in\mathcal{Z}.
\]
Thus, $\mu_0(z)=g(z)F_0(z)=1/4$ at every grid point, and no action has a larger mean reward; see Figure~\ref{fig:lower-bound-instances}(a).
Precisely, let $j\in\{0,\ldots,M-1\}^d$ index a perturbed instance, with associated cell $A_j=\prod_i[\rho j_i,\rho(j_i+1))$ and lower corner $z_j$. Its vertices are the grid points $z_{j+s}$, with $s\in\{0,1\}^d$. Relative to the reference distribution, we increase the probability mass at each vertex by $\varepsilon$ if $\sum_i s_i$ is even and decrease it by $\varepsilon$ if it is odd. All other probabilities remain unchanged. The increases and decreases cancel outside $A_j$, so the feedback has the same distribution under both instances there. Inside $A_j$, the probability of trade increases by $\varepsilon$, and the mean reward at $z_j$ increases by at least $\varepsilon/2$, making $z_j$ optimal; see Figure~\ref{fig:lower-bound-instances}(b).

For any action $x\in A_j$, the KL divergence between the feedback distributions under the reference and perturbed instances equals that between two Bernoulli distributions with parameters $F_0(z_j)$ and $F_0(z_j)+\varepsilon$, respectively. Since these probabilities are bounded away from zero and one, this divergence is $O(\varepsilon^2)$.

The $M^d$ disjoint cells thus play a role similar to $M^d$ arms in a bandit lower bound: each alternative improves one cell by order $\varepsilon$, but can only be distinguished from the reference instance by exploring that cell.

For a sufficiently small constant $\gamma_d>0$, choosing $\varepsilon=\gamma_d\min\{M^{-d},\sqrt{M^d/T}\}$ keeps these distributions well-defined. A change-of-measure argument gives a regret lower bound of order
$\min\left\{\frac{T}{M^d},\sqrt{M^dT}\right\}$
Optimizing $M$ subject to each marginal support containing at most $K$ points yields Theorem~\ref{thm:lower-bound}. %The full proof is in Appendix~\ref{app:lower-bound-proof}.
\label{sec:lower-bound-end}

\bibliography{references}

\clearpage
\appendix
\section*{Supplementary Materials}

\section{Omitted Proofs}

\LemmaPalla*

\begin{proof}
Suppose $y=y(x)\in G$ exists and is such that $F(y)>0$. In this case we have that $F(x)=F(y)>0$ and by hypothesis $\lVert x-y\rVert_{1,\alpha}>\varepsilon$.
Hence, 
\begin{align*}
    \Delta(x)&\ge \mu(y)-\mu(x)= (g(y)-g(x))F(x) = \lVert x-y\rVert_{1,\alpha} F(x) >  \varepsilon F(x) .
\end{align*}

Otherwise, if $y(x)$ is not well-defined or is such that $F(y)=0$, then, this implies $F(x)=0$, and the statement becomes $\Delta(x)\ge 0$, which is true by the definition of the suboptimality gap.
\end{proof}

\LemmaClean*

\begin{proof}
For each fixed cell and sample count $1\le n\le T$, Hoeffding's inequality bounds the failure probability 
\[\Pr\{|\widehat F_{A,n}-F(u_A)|>\beta_{A,n}\}\le2e^{-2c}\le\frac{\delta}{N_HT}.\]
%The observations assigned to a cell are fresh Bernoulli draws at a fixed corner: one may equivalently assign each possible cell an independent infinite sample stream and reveal the next entry when it is played. 
Thus, a union bound over cells and sample counts proves that
%is valid despite adaptive sampling. Consequently 
$\Pr(\mathcal E)\ge1-\delta$.
\end{proof}

\lemmaExclude*

\begin{proof}
Let $A$ be the cell selected at a certain round, and assume that $u_A\notin \bigcup_{y\in G}B^{\uparrow,\alpha}_{6\omega_A}(y)$ and the depth of $A$ is smaller than $H$. Since $A$ is chosen it holds that
\begin{equation}\label{eq:optimality}
g(\ell_A) F^+_A \ge g(\ell_B) F^+_B  \quad \forall B \in \mathcal P  
\end{equation}
and we suppose by contradiction that \begin{equation}
     F_A^+> \max \left\{4,\frac{g(\ell_A)}{g(\ell_A)-g(u_A)}\right\}\beta_A. \label{eq:no-split}
\end{equation}
Then, if $\frac{g(\ell_A)}{g(\ell_A)-g(u_A)}\ge4$
\begin{equation*}%\label{eq:atleasthalf}
F(u_A) \ge F_A^+ - 2 \beta_A > F_A^+\left(1-\frac{2 \omega_A}{g(\ell_A)}      \right) \ge \frac{F_A^+}{2}>0.
\end{equation*}
and if $\frac{g(\ell_A)}{g(\ell_A)-g(u_A)}<4$
\[F(u_A)\ge F_A^+-2\beta_A> \frac{1}{2}F_A^+. \]

Therefore, in both cases, if by contradiction we assume that $A$ is split it must be true that
\begin{equation}\label{eq:atleasthalf}
    F(u_A)> \frac{1}{2}F_A^+.
\end{equation}

Moreover, let $B^*$ be the cell active at current round which includes the optimum. By \eqref{eq:optimism} and \eqref{eq:optimality}, we get
\[ g(\ell_A) F^+_A \ge g(\ell_{B^*}) F^+_{B^*}\ge \max_{x\in B^*}\, \mu(x)= \mu^\star\]

Hence, we get:
\begin{align}
 \Delta(u_A)
 &\le g(\ell_A)F_A^+-g(u_A)F(u_A)\notag\\
 &=g(\ell_A)\left(F_A^+- F(u_A)\right)+ \left(g(\ell_A)-g(u_A)\right)F(u_A)\notag\\
 &\le 2 \beta_A g(\ell_A)+\omega_A F(u_A)\notag\\
 & <2 \omega_A F^+_A+\omega_A F(u_A)\notag\\
 &\le3 \omega_A F^+_A <6 \omega_AF(u_A)\label{eq:split-gap}
\end{align}
where in the last inequality we used \eqref{eq:atleasthalf}.
However, by hypothesis we have
\[u_A\notin  B^{\uparrow,\alpha}_{6\omega_A}(y) \quad \forall y \in G\]
which means by Lemma \ref{lem:gap}
\[ \Delta(u_A)> 6\omega_A F(u_A) .\]
Hence, we reach a contradiction.
\end{proof}

\lemmaNumber*

\begin{proof}
First, we bound the number of cell that are split at a level $h$.
Fix $h<H$, and let us recall that at depth $h$ all cells of the full tree have the same side-length vector. Let $k_i(h)$ be the number of bisections in coordinate $i$ along a depth-$h$ branch. %These counts are independent of the branch, since the choice of split coordinate depends only on the side lengths. Write
%\begin{align*}
% w_i(h)&=\alpha_i2^{-k_i(h)},\\
%  \omega_h&=\sum_i w_i(h).
%\end{align*}
After at most $d$ bisections every weighted side is at most half the previous maximum: each side larger than that threshold is halved before any smaller side is selected. Since $\omega_0\le 2$,
\begin{equation*}%\label{eq:depth-error}
 \omega_h\le \,2^{-\lfloor h/d\rfloor+1}.
\end{equation*}

Every splitting cell belongs to one of two classes; $g(\ell_A)\ge0 $ or $g(\ell_A)<0$.

\begin{itemize}
    \item $\mathbf{g(\ell_A)\ge0 }.$
    If $g(\ell_A)\ge 0 $, Lemma~\ref{lem:exclude} shows that its upper corner belongs to $B^{\uparrow,\alpha}_{6\omega_A}(y)$ for some $y\in G$. 
\[
 \sum_{i=1}^d\alpha_i (u_{A,i}-y_i)\le 6\omega_A\le 6\omega_h.
\]
Therefore, any cell such that $g(\ell_A)\ge 0$ at level $h$ that wants to be split must guarantee for any $i$ 
\[y_i\le  u_{A,i}\le y_i + \frac{6\omega_h}{\alpha_i}.\]
If the direction $i$ was never split, then trivially there is only one possible value for $u_{A,i}$.
Let $r_{h,i}=2^{-k_i(h)}$ be the geometric length of the $i$-th coordinate of a cell at depth $h$.
If coordinate $i$ has previously been bisected, let $h^\prime$ be the last depth for which coordinate $i$ was split, i.e. 
\[\alpha_i r_{h',i}\ge \alpha_j r_{h',j} \qquad \forall j\neq i,\]
 then 
\begin{align*}
    2\alpha_i r_{h,i}&=\alpha_i r_{h',i}\ge \max_{j}\alpha_jr_{h',j}\ge \frac{1}{d}\sum_{j=1}^d\alpha_jr_{h',j}
    \ge \frac{1}{d}\omega_{h'}\ge \frac{1}{d}\omega_{h}
\end{align*}  
as the diameter of a cell in norm $\|\cdot \|_{1,\alpha}$ can only decrease when the depth of the tree increases.
Therefore we have that 
\[r_{h,i}\ge \frac{\omega_h}{2d\alpha_i},\]
which means that for all $y\in G$ there are at most $12d+2$ possible choices for $u_{A,i}$ for all $i$ that lead to a splitting. Considering that each cell is split in two, summing over all $y\in G$, we get that at level $h+1$ there are at most $2(12d+2)^d(|G|)$ active cells, generated by cells at level $h$ with $g(\ell_A)\ge 0$.

\item  $\mathbf{g(\ell_A)<0}.$
No cell A such that $g(\ell_A)<0$ will ever be selected and therefore split. In fact, if $A$ is selected this implies 
\[g(\ell_A)F_A^+\ge \max_{B\in \mathcal{P}}g(\ell_B)F_B^+\ge \mu^\star,\]
however, since $F_A^+>0$, then $g(\ell_A)F_A^+<0$, which is a contradiction since there exists $x \in [0,1]^d$ such that $g(x)\ge 0$, since $\mu^\star\ge g(x)F(x)\ge 0$.
\end{itemize}

Therefore, at each level $h\le H$ there are at most \(2(12d+2)^d|G|\) cells, which concludes the proof.
\end{proof}

\section{Proof of Theorem~\ref{thm:lower-bound}}
\label{app:lower-bound-proof}

We introduce a reference distribution whose support consists of an equally spaced grid in $[0,1/2]^d$ and the all-ones vector $\mathbf{1}=(1,\ldots,1)$. Fix integers $d\ge1$ and $M\ge2$, and set
\[
\rho=\frac{1}{2M},\qquad
z_j=\rho j,\quad j\in\{0,\ldots,M\}^d.
\]
The parameter $M$ is the number of grid intervals along each coordinate.
We define the set $\mathcal{V}$ as follows
\[
\mathcal Z=\{z_j:j\in\{0,\ldots,M\}^d\},
\qquad
\mathcal V=\mathcal Z\cup\{\mathbf1\}.
\]
Since $M\rho=1/2$, every point of the grid $\mathcal Z$ lies in $[0,1/2]^d$, and the additional point $\mathbf{1}$ is distinct from every grid point.
We define
\[
g(x)=1-\frac1d\sum_{i=1}^d x_i,
\qquad x\in[0,1]^d,
\]
and the auxiliary function $H(x)=\frac{1}{4g(x)}$ for every $x\in[0,1/2]^d$.
On $[0,1/2]^d$, we have $1/2\le g(x)\le1$, so $H$ is well defined and smooth.
We define our reference discrete probability distribution by assigning the following probabilities:
\begin{equation}
	\label{eq:mass-definition}
	p(x)=
	\begin{cases}
		\displaystyle
		\sum_{{s\in\{0,1\}^d \, : \, s\preceq j}}
		(-1)^{\sum_{i=1}^d s_i}H(z_{j-s}),
		& \text{if } x=z_j\in\mathcal Z,\\[1.2em]
		1/2,
		& \text{if } x=\mathbf1,\\[0.4em]
		0,
		& \text{otherwise}.
	\end{cases}
\end{equation}
Lemma~\ref{lem:validity} shows that $p$, defined in Equation~\eqref{eq:mass-definition}, defines a discrete probability distribution supported on $\mathcal{V}$. We denote its CDF by $F_0$. In Lemma~\ref{lem:cdf}, we show that, for every $z_k\in\mathcal Z$,
\begin{equation}
	\label{eq:cdf-identity}
	F_0(z_k)=H(z_k)=\frac{1}{4g(z_k)}.
\end{equation}
As a consequence, letting $\mu_0(x):=g(x)F_0(x)$, every point of the grid $\mathcal Z$ has the same expected reward:
\[
\mu_0(z_k)=g(z_k)F_0(z_k)
=g(z_k)\frac{1}{4g(z_k)}=\frac14.
\]

For each $j\in\{0,\ldots,M-1\}^d$, consider the cell
\[
A_j=\prod_{i=1}^d[\rho j_i,\rho(j_i+1)).
\]
These $M^d$ cells are pairwise disjoint and partition $[0,1/2)^d$.
The vertices of $A_j$ are the grid points $z_{j+s}$, with
$s\in\{0,1\}^d$. Fix
\begin{equation}
\label{eq:epsilon-range}
0<\varepsilon\le\frac{d!}{8d^d}\rho^d.
\end{equation}
We define the perturbed discrete probability distribution associated with $A_j$ by
\begin{equation}
\label{eq:perturbed-mass}
p_j(x)=
\begin{cases}
p(x)+\varepsilon(-1)^{\sum_{i=1}^d s_i},
&\text{if }x=z_{j+s}\text{ for some }s\in\{0,1\}^d,\\[0.4em]
p(x),&\text{otherwise}.
\end{cases}
\end{equation}
Thus we add $\varepsilon$ at vertices for which $\sum_i s_i$ is even
and subtract $\varepsilon$ at vertices for which it is odd.
In particular, the probability assigned to $z_j$ increases by $\varepsilon$, and the
probability assigned to $\mathbf1$ remains $1/2$.
Lemma~\ref{lem:perturbed-cdf} shows that $p_j$ defines a discrete probability
distribution with support $\mathcal V$. Denoting its CDF by $F_j$, the same
lemma shows that $F_j(x)$ exceeds $F_0(x)$ by exactly $\varepsilon$
on $A_j$ and coincides with $F_0(x)$ everywhere else in $[0,1]^d$.
Formally, 
\[
F_j(x):=F_0(x)+\varepsilon \mathbf1\{x\in A_j\}.
\]
Consequently, the expected reward under the perturbed instance is
\[
\mu_j(x):=g(x)F_j(x)
=\mu_0(x)+\varepsilon g(x)\mathbf1\{x\in A_j\}.
\]
In particular, since $z_j\in A_j$ and $g(z_j)\ge1/2$,
\[
\mu_j(z_j)=\frac14+\varepsilon g(z_j)
\ge\frac14+\frac{\varepsilon}{2}.
\]
Moreover, $\mu_0(x)\le1/4$ for every $x\in[0,1]^d$.
Indeed, if $x\ne\mathbf1$, let $z_k$ be the largest grid point
coordinatewise dominated by $x$. The grid points dominated by $x$
are exactly those dominated by $z_k$. Hence $F_0(x)=F_0(z_k)=1/(4g(z_k))$.
Since $g(x)\le g(z_k)$, we obtain $\mu_0(x)\le1/4$.
At $x=\mathbf1$, the expected reward is zero because $g(\mathbf1)=0$.
It follows that every action outside $A_j$ has expected reward at
most $1/4$ under instance $j$, whereas $z_j$ has expected reward at
least $1/4+\varepsilon/2$.
Thus each perturbed instance improves the reward in a single cell,
while leaving the probability of trade unchanged outside that cell.

Set $N=M^d$, and let $\mathbb P_0$ and $\mathbb P_j$ denote
the laws of the $T$-round interaction between the possibly
randomized algorithm and the stochastic environment, under
the reference instance and instance $j$, respectively.
These laws account for both the algorithm's internal randomness
and the randomness of the feedback. Expectations
under these laws are denoted by $\mathbb E_0$ and $\mathbb E_j$.
All instances share the support $\mathcal V$, with marginal
support sets $\{0,\rho,\ldots,M\rho,1\}$ of size $M+2$.
We prove the lower bound even when the learner knows this
common support, but not which instance $j$ generates the observations. Define
\(
N_j(T)=\sum_{t=1}^T\mathbf1\{x_t\in A_j\}.
\)
The gap outside $A_j$ is at least $\varepsilon/2$, so
\begin{equation}
\label{eq:regret-from-visits}
\mathrm{Reg}_{T,j}\ge
\frac{\varepsilon}{2}\bigl(T-\mathbb E_jN_j(T)\bigr).
\end{equation}
Here $\mathrm{Reg}_{T,j}$ denotes the learner's expected
regret under instance $j$.
Lemma~\ref{lem:kl-comparison} gives
\begin{equation}
\label{eq:visits-change-of-measure}
\mathbb E_jN_j(T)
\le\mathbb E_0N_j(T)+2\varepsilon T\sqrt{\mathbb E_0N_j(T)}.
\end{equation}
The cells are disjoint, hence
$\sum_{j\in\{0,\ldots,M-1\}^d}\mathbb E_0N_j(T)\le T$.
Averaging Equation~\eqref{eq:visits-change-of-measure} over the $N$
instances and applying Cauchy-Schwarz yields
\begin{equation}
\label{eq:average-visits}
\frac1N\sum_{j\in\{0,\ldots,M-1\}^d}\mathbb E_jN_j(T)
\le\frac TN+2\varepsilon T\sqrt{\frac TN}.
\end{equation}
Let $\gamma_d=d!/[8(2d)^d]$.
Since $d\ge1$, we have $\gamma_d\le1/16$. 
We choose
\(
\varepsilon=\gamma_d\min\{1/N,\sqrt{ N/T}\}.
\)
This choice satisfies Equation~\eqref{eq:epsilon-range}, since
\[
0<\varepsilon\le\frac{\gamma_d}{N}
=\frac{d!}{8d^d}\rho^d.
\]
Averaging over all $N=M^d$ perturbed instances, we obtain
\begin{align*}
\frac1N\sum_{j\in\{0,\ldots,M-1\}^d}\mathrm{Reg}_{T,j}
&\ge\frac{\varepsilon}{2}\left(
T-\frac1N\sum_{j\in\{0,\ldots,M-1\}^d}\mathbb E_jN_j(T)\right)\\
&\ge\frac{\varepsilon}{2}\left(
T-\frac TN-2\varepsilon T\sqrt{\frac TN}\right)\\
&\ge\frac{\varepsilon}{2}\left(T-\frac{3T}{4}\right)
=\frac{\varepsilon T}{8}.
\end{align*}
The first inequality follows from Equation~\eqref{eq:regret-from-visits},
the second from Equation~\eqref{eq:average-visits}, and the third uses
$N\ge2$, $\varepsilon\le\gamma_d\sqrt{N/T}$, and $\gamma_d\le1/16$.
At least one instance has regret no smaller than this average, so some
$j\in\{0,\ldots,M-1\}^d$ satisfies
\begin{equation}
\label{eq:finite-grid-regret}
\mathrm{Reg}_{T,j}\ge\frac{\varepsilon T}{8}
=\frac{\gamma_d}{8}\min\left\{\frac{T}{M^d},\sqrt{M^dT}\right\}.
\end{equation}

For $K\ge4$ and $T\ge2^{3d}$, choose
$M=\min\{K-2,\lfloor T^{1/(3d)}\rfloor\}$.
This choice satisfies $2\le M\le K-2$, as required by the construction,
and
\(
M^{3d}\le T \) and  
\(M\ge\frac12\min\{K,T^{1/(3d)}\}.
\)
For the instance $j$ in Equation~\eqref{eq:finite-grid-regret},
we obtain
\[
\begin{aligned}
\mathrm{Reg}_T=\mathrm{Reg}_{T,j}
&\ge\frac{\varepsilon T}{8}\\
&=\frac{\gamma_d}{8}
\min\left\{\frac{T}{M^d},\sqrt{M^dT}\right\}\\
&=\frac{\gamma_d}{8}\sqrt{M^dT}\\
&\ge\frac{\gamma_d}{8\,2^{d/2}}\sqrt T\,
\bigl(\min\{K,T^{1/(3d)}\}\bigr)^{d/2}\\
&=\frac{\gamma_d}{8\,2^{d/2}}
\min\{\sqrt{K^dT},T^{2/3}\}.
\end{aligned}
\]
The first inequality follows from Equation~\eqref{eq:finite-grid-regret}.
The minimum equals $\sqrt{M^dT}$ because $M^{3d}\le T$, and the
second inequality uses $M\ge\frac12\min\{K,T^{1/(3d)}\}$.
\par
If $K\ge4$ and $T<2^{3d}$, take $M=2$ instead.
Then Equation~\eqref{eq:finite-grid-regret} gives
$\mathrm{Reg}_T\ge\gamma_d T/(8\,2^d)$, which is at least
$\gamma_d\min\{\sqrt{K^dT},T^{2/3}\}/(8\,2^d)$ because $T\ge1$.

\subsection{Proof of the Lemmas} \label{sec:useful-lemmas}
\begin{lemma}
\label{lem:cancellation}
Let $\xi,\zeta\in\mathbb Z_{\ge0}^d$ satisfy $\xi\preceq \zeta$, and let $w>0$.
Then
\begin{equation}
\label{eq:weighted-factorization}
\begin{aligned}
\sum_{\substack{\eta\in\{0,1\}^d\\\xi+\eta\preceq \zeta}}
(-w)^{\sum_{i=1}^d \eta_i}
=\prod_{i=1}^d\left(
\sum_{\substack{\eta_i\in\{0,1\}\\\xi_i+\eta_i\le \zeta_i}}
(-w)^{\eta_i}\right)=\prod_{i=1}^d\bigl(1-w\mathbf1\{\xi_i<\zeta_i\}\bigr).
\end{aligned}
\end{equation}
\end{lemma}

\begin{proof}
Since the constraint $\xi+\eta\preceq \zeta$ is coordinatewise, we can write
\begin{align*}
\sum_{\substack{\eta\in\{0,1\}^d\\\xi+\eta\preceq \zeta}}
(-w)^{\sum_{i=1}^d \eta_i}
&=\sum_{\substack{\eta_1\in\{0,1\}\\\xi_1+\eta_1\le \zeta_1}}
\cdots
\sum_{\substack{\eta_d\in\{0,1\}\\\xi_d+\eta_d\le \zeta_d}}
(-w)^{\eta_1+\cdots+\eta_d}\\
&=\left(\sum_{\substack{\eta_1\in\{0,1\}\\\xi_1+\eta_1\le \zeta_1}}
(-w)^{\eta_1}\right)\cdots
\left(\sum_{\substack{\eta_d\in\{0,1\}\\\xi_d+\eta_d\le \zeta_d}}
(-w)^{\eta_d}\right)\\
&=\prod_{i=1}^d\left(
\sum_{\substack{\eta_i\in\{0,1\}\\\xi_i+\eta_i\le \zeta_i}}
(-w)^{\eta_i}\right).
\end{align*}
If $\xi_i=\zeta_i$, only $\eta_i=0$ is admissible; if $\xi_i<\zeta_i$, both
$\eta_i=0$ and $\eta_i=1$ are admissible, since $\xi_i$ and $\zeta_i$ are integers.
Thus each factor satisfies
\[
\sum_{\substack{\eta_i\in\{0,1\}\\\xi_i+\eta_i\le \zeta_i}}(-w)^{\eta_i}
=\begin{cases}
1,&\xi_i=\zeta_i,\\
1-w,&\xi_i<\zeta_i.
\end{cases}
\]
This proves Equation~\eqref{eq:weighted-factorization}.
\end{proof}

\begin{lemma}
\label{lem:grid-masses}
For every $j\in\{0,\ldots,M\}^d$, let
$r(j)=|\{i\in\{1,\ldots,d\}:j_i>0\}|$ be the number of its positive
coordinates. The probabilities defined in Equation~\eqref{eq:mass-definition}
satisfy
\begin{equation}
\label{eq:mass-integral}
p(z_j)=\frac14\int_0^1
 t^{g(z_j)-1}(1-t^{\rho/d})^{r(j)}\,dt.
 \end{equation}
 Furthermore, the following holds:
 \begin{equation}
 	\label{eq:mass-integral-1}
 	p(z_j)\ge\frac{r(j)!}{4d^{r(j)}}\rho^{r(j)}
\ge\frac{d!}{4d^d}\rho^d>0.
 \end{equation}
\end{lemma}

\begin{proof}
We use the elementary identity
\begin{equation}
\label{eq:reciprocal-integral}
\int_0^1 t^{a-1}\,dt
=\frac1a,\qquad a>0.
\end{equation}
For every $s\in\{0,1\}^d$ with $s\preceq j$,
\begin{equation}
\label{eq:shifted-g}
g(z_{j-s})
=1-\frac{\rho}{d}\sum_{i=1}^d(j_i-s_i)
=g(z_j)+\frac{\rho}{d}\sum_{i=1}^d s_i.
\end{equation}
Since $g(z_{j-s})\ge1/2>0$, we can apply
Equation~\eqref{eq:reciprocal-integral} with $a=g(z_{j-s})$.
Using Equation~\eqref{eq:shifted-g}, we obtain
\begin{equation}
\label{eq:h-integral}
\begin{aligned}
H(z_{j-s})=\frac{1}{4g(z_{j-s})}
=\frac14\int_0^1 t^{g(z_{j-s})-1}\,dt=\frac14\int_0^1
 t^{g(z_j)-1}t^{(\rho/d)\sum_i s_i}\,dt.
\end{aligned}
\end{equation}
Substituting Equation~\eqref{eq:h-integral} into
Equation~\eqref{eq:mass-definition} and exchanging the finite sum
with the integral, we obtain
\[
p(z_j)=\frac14\int_0^1 t^{g(z_j)-1}
\sum_{\substack{s\in\{0,1\}^d\\s\preceq j}}
(-1)^{\sum_i s_i}t^{(\rho/d)\sum_i s_i}\,dt.
\]
For every fixed $t\in(0,1)$, apply Lemma~\ref{lem:cancellation}
with $\xi=\mathbf0$, $\zeta=j$, and $w=t^{\rho/d}>0$. This gives
\begin{align*}
\sum_{\substack{s\in\{0,1\}^d\\s\preceq j}}
(-1)^{\sum_i s_i}t^{(\rho/d)\sum_i s_i}
=\sum_{\substack{s\in\{0,1\}^d\\s\preceq j}}
(-t^{\rho/d})^{\sum_i s_i}
=\prod_{i=1}^d\bigl(1-t^{\rho/d}\mathbf1\{j_i>0\}\bigr)
=(1-t^{\rho/d})^{r(j)}.
\end{align*}
The last equality holds because exactly $r(j)$ coordinates of $j$ are
positive. Substituting this expression into the integral proves Equation~\eqref{eq:mass-integral}.

In Equation~\eqref{eq:mass-integral}, set $t=e^{-u}$.
Then $dt=-e^{-u}\,du$, and the endpoints $t=0,1$ correspond to
$u=\infty,0$, respectively. Reversing the integration limits gives
\[
p(z_j)=\frac14\int_0^\infty
 e^{-g(z_j)u}(1-e^{-(\rho/d)u})^{r(j)}\,du.
\]
For every $u\ge0$,
\[
1-e^{-(\rho/d)u}
=\int_0^{(\rho/d)u}e^{-v}\,dv
\ge\frac{\rho}{d}u\,e^{-(\rho/d)u}.
\]
Moreover, $\sum_i j_i\ge r(j)$, so
\[
g(z_j)+\frac{r(j)\rho}{d}
=1-\frac{\rho}{d}\left(\sum_i j_i-r(j)\right)\le1.
\]
Consequently,
\begin{align*}
p(z_j)
&\ge\frac14\left(\frac{\rho}{d}\right)^{r(j)}
\int_0^\infty u^{r(j)} e^{-(g(z_j)+r(j)\rho/d)u}\,du \ge\frac14\left(\frac{\rho}{d}\right)^{r(j)}
\int_0^\infty u^{r(j)} e^{-u}\,du
=\frac{r(j)!}{4d^{r(j)}}\rho^{r(j)}.
\end{align*}
The last integral is $1$ when $r(j)=0$. For $r(j)\ge1$,
integration by parts gives
\[
\int_0^\infty u^{r(j)} e^{-u}\,du
=\bigl[-u^{r(j)} e^{-u}\bigr]_0^\infty
+r(j)\int_0^\infty u^{r(j)-1}e^{-u}\,du
=r(j)\int_0^\infty u^{r(j)-1}e^{-u}\,du.
\]
Repeating this integration-by-parts step $r(j)$ times and using
$\int_0^\infty e^{-u}\,du=1$, we obtain
\[
\int_0^\infty u^{r(j)}e^{-u}\,du=r(j)!.
\]
This proves the first inequality in Equation~\eqref{eq:mass-integral-1}. Finally, since $r(j)\le d$ and $\rho\le1$,
\[
\frac{d!\rho^d/d^d}{r(j)!\rho^{r(j)}/d^{r(j)}}
=\prod_{m=r(j)+1}^d\frac{m\rho}{d}\le1,
\]
where an empty product equals one. Hence, for every grid point,
\begin{equation}
\label{eq:grid-mass-lower-bound}
p(z_j)\ge\frac{r(j)!}{4d^{r(j)}}\rho^{r(j)}
\ge\frac{d!}{4d^d}\rho^d,
\end{equation}
concluding the proof.
\end{proof}

\begin{lemma}
\label{lem:validity}
The function $p$ defined in Equation~\eqref{eq:mass-definition} defines a discrete probability distribution. In particular,
\[
p(x)\ge0\quad \forall x\in\mathcal V ,
\qquad
\sum_{x\in\mathcal V}p(x)=1.
\]
Moreover, $p(x)>0$ at every $x\in\mathcal V$, so its support is
exactly $\mathcal V$.
\end{lemma}

\begin{proof}
By Lemma~\ref{lem:grid-masses}, $p(z_j)>0$ at every grid point.
Also, $p(\mathbf1)=1/2>0$, and $p(x)=0$ outside $\mathcal V$.
It therefore remains to show that the probabilities sum to one.
We use Lemma~\ref{lem:cancellation} to compute the total probability assigned to the grid. Summing Equation~\eqref{eq:mass-definition} over $\mathcal Z$ gives
\[
\sum_{j\in\{0,\ldots,M\}^d}p(z_j)
=
\sum_{j\in\{0,\ldots,M\}^d}
\sum_{\substack{s\in\{0,1\}^d\\s\preceq j}}
(-1)^{\sum_i s_i}H(z_{j-s}).
\]
Set $k=j-s$, or equivalently $j=k+s$. The original constraints on
$(j,s)$ are equivalent to
\(
s\in\{0,1\}^d, \mathbf0\preceq k,
 k+s\preceq M\mathbf1.
\)
Therefore, we can write:
\begin{align*}
\sum_{j\in\{0,\ldots,M\}^d}p(z_j)
= \hspace{-3mm}
\sum_{k\in\{0,\ldots,M\}^d}H(z_k)
\sum_{\substack{s\in\{0,1\}^d\\k+s\preceq M\mathbf1}}
(-1)^{\sum_i s_i}
= \hspace{-4mm}
\sum_{k\in\{0,\ldots,M\}^d}  \hspace{-4mm} H(z_k) 
\mathbf1\{k=M\mathbf1\}
=H(z_{M\mathbf1}).
\end{align*}
The second equality follows from Lemma~\ref{lem:cancellation} with $\xi=k$, $\zeta=M\mathbf1$, and $w=1$. Since $M\rho=1/2$,
\[
H(z_{M\mathbf1})
=H\!\left(\tfrac12\mathbf1\right)
=\frac{1}{4\left(1-\frac1d\sum_{i=1}^d\frac12\right)}
=\frac{1}{4(1-1/2)}
=\frac12.
\]
Finally, the additional point $\mathbf1$ lies outside the small grid
and has probability $1/2$.  This proves that $p$ defines a discrete probability distribution. Its support is
exactly $\mathcal V$, {because every point in $\mathcal V$ has strictly}
positive probability and all other points have probability zero.
\end{proof}

\begin{lemma}
\label{lem:cdf}
For every $z_k\in\mathcal Z$,
\begin{equation}
\label{eq:cdf-identity-lemma}
F_0(z_k)=H(z_k)=\frac{1}{4g(z_k)}.
\end{equation}
\end{lemma}

\begin{proof}
Fix $k\in\{0,\ldots,M\}^d$. The point $\mathbf1$ does not contribute
to $F_0(z_k)$, because $z_k\in[0,1/2]^d$ and hence
$\mathbf1\npreceq z_k$. Also, since $\rho>0$, the condition
$z_j\preceq z_k$ is equivalent to $j\preceq k$. Therefore,
\begin{align*}
F_0(z_k)
=\sum_{\mathbf0\preceq j\preceq k}p(z_j)
=\sum_{\mathbf0\preceq j\preceq k}
\sum_{\substack{s\in\{0,1\}^d\\s\preceq j}}
(-1)^{\sum_i s_i}H(z_{j-s}).
\end{align*}
Set $\ell=j-s$, so that $j=\ell+s$. The constraints become
\(
\ell\succeq\mathbf0, s\in\{0,1\}^d, \ell+s\preceq k.
\)
For each $\ell$ with $\mathbf0\preceq\ell\preceq k$,
we apply Lemma~\ref{lem:cancellation} with $\xi=\ell$, $\zeta=k$,
and $w=1$. This gives
\begin{align*}
F_0(z_k)
=\sum_{\mathbf0\preceq\ell\preceq k}H(z_\ell)
\sum_{\substack{s\in\{0,1\}^d\\\ell+s\preceq k}}
(-1)^{\sum_i s_i} =\sum_{\mathbf0\preceq\ell\preceq k}H(z_\ell)\mathbf1\{\ell=k\}
=H(z_k)=\frac{1}{4g(z_k)},
\end{align*}
concluding the proof.
\end{proof}

\begin{lemma}
\label{lem:perturbed-cdf}
For every $j\in\{0,\ldots,M-1\}^d$ and every $\varepsilon$ satisfying
Equation~\eqref{eq:epsilon-range}, the function $p_j$ defined in
Equation~\eqref{eq:perturbed-mass} defines a discrete probability distribution with support
$\mathcal V$. Its CDF satisfies
\begin{equation}
\label{eq:localized-cdf}
F_j(x)-F_0(x)=\varepsilon\mathbf1\{x\in A_j\},
\qquad\text{for all }x\in[0,1]^d.
\end{equation}
\end{lemma}

\begin{proof}
By Equations~\eqref{eq:mass-integral-1} and~\eqref{eq:epsilon-range},
\[
p(z_k)\ge\frac{d!}{4d^d}\rho^d\ge2\varepsilon
\qquad\text{for every }k\in\{0,\ldots,M\}^d.
\]
Thus $p(z_k)-\varepsilon\ge\varepsilon>0$ at every grid point.
Every probability assigned by $p_j$ to a point of $\mathcal V$ is
therefore positive, and $p_j(x)=0$ outside $\mathcal V$.
Since $p_j=p$ away from the vertices of $A_j$, we have
\[
\begin{aligned}
\sum_{x\in\mathcal V}p_j(x)
&=\sum_{x\in\mathcal V}p(x)
 +\sum_{s\in\{0,1\}^d}
 \bigl(p_j(z_{j+s})-p(z_{j+s})\bigr)\\
&=1+\varepsilon\sum_{s\in\{0,1\}^d}(-1)^{\sum_i s_i}=1.
\end{aligned}
\]
The second equality uses the normalization of $p$ and the definition
of $p_j$. The last equality follows from Lemma~\ref{lem:cancellation}
with $\xi=\mathbf0$, $\zeta=\mathbf1$, and $w=1$.
Hence $p_j$ defines a discrete probability distribution with support
exactly $\mathcal V$.
Fix $x\in[0,1]^d$. Only the probabilities at the vertices of $A_j$ have been modified,
so, by the definition of a CDF,
\[
\begin{aligned}
F_j(x)-F_0(x)
&=\sum_{\substack{v\in\mathcal V\\v\preceq x}}
 \bigl(p_j(v)-p(v)\bigr)\\
&=\sum_{\substack{s\in\{0,1\}^d\\z_{j+s}\preceq x}}
 \bigl(p_j(z_{j+s})-p(z_{j+s})\bigr)\\
&=\varepsilon\sum_{\substack{s\in\{0,1\}^d\\z_{j+s}\preceq x}}
 (-1)^{\sum_i s_i},
\end{aligned}
\]
where the last equality uses Equation~\eqref{eq:perturbed-mass}.
\par
If $z_j\npreceq x$, no vertex $z_{j+s}$ is dominated by $x$, where $s \in \{0,1\}^d$.
The sum is then empty and equals zero, and $x\notin A_j$.
Otherwise, set $b_i=\lfloor x_i/\rho\rfloor$ for each coordinate $i$.
Then $b\in\mathbb Z_{\ge0}^d$, $j\preceq b$, and
\[
z_{j+s}\preceq x
\quad\Longleftrightarrow\quad
j+s\preceq b,
\]
Applying Lemma~\ref{lem:cancellation} with $\xi=j$,
$\zeta=b$, and $w=1$, we obtain
\[
F_j(x)-F_0(x)
=\varepsilon\sum_{\substack{s\in\{0,1\}^d\\j+s\preceq b}}
(-1)^{\sum_i s_i}
=\varepsilon\mathbf1\{j=b\}.
\]
Finally,
\[
b=j
\quad\Longleftrightarrow\quad
\rho j_i\le x_i<\rho(j_i+1)\ \text{for every }i
\quad\Longleftrightarrow\quad x\in A_j.
\]
Together with the case $z_j\npreceq x$, this proves
Equation~\eqref{eq:localized-cdf} for every $x\in[0,1]^d$.
\end{proof}

\begin{lemma}
\label{lem:kl-comparison}
For every $j\in\{0,\ldots,M-1\}^d$ and every
$\varepsilon$ satisfying Equation~\eqref{eq:epsilon-range},
\begin{equation}
\label{eq:cell-history-kl}
\mathbb E_jN_j(T)\le\mathbb E_0N_j(T)
+2\varepsilon T\sqrt{\mathbb E_0N_j(T)}.
\end{equation}
\end{lemma}

\begin{proof}
Let $\operatorname{kl}(a,b)$ denote the KL divergence between Bernoulli
laws with parameters $a,b$, using natural logarithms.
The divergence decomposition of Lattimore and Szepesv\'ari~\cite[Lemma~15.1]{lattimore2020bandit}
gives
\begin{equation}
\label{eq:adaptive-kl}
\operatorname{KL}(\mathbb P_0,\mathbb P_j)
=\mathbb E_0\!\left[
\sum_{t=1}^T\operatorname{kl}(F_0(x_t),F_j(x_t))\right].
\end{equation}
Furthermore, for any random variable $W$ measurable with respect to the interaction history and satisfying $0\le W\le T$, we have
\begin{equation}
\label{eq:bounded-expectations}
\begin{aligned}
\left|\mathbb E_jW-\mathbb E_0W\right|
&=\left|\int_0^T
\bigl[\mathbb P_j(W>u)-\mathbb P_0(W>u)\bigr]\,du\right|\\
&\le\int_0^T
\left|\mathbb P_j(W>u)-\mathbb P_0(W>u)\right|\,du\\
&\le\int_0^T\sqrt{\tfrac12\operatorname{KL}(\mathbb P_0,\mathbb P_j)}\,du\\
&=T\sqrt{\tfrac12\operatorname{KL}(\mathbb P_0,\mathbb P_j)}.
\end{aligned}
\end{equation}
The last inequality follows from Equation~(14.12) in ~\cite{lattimore2020bandit}.
In addition, the KL divergence
between two Bernoulli distributions can be upper bounded as follows
\begin{equation}
\label{eq:bernoulli-kl-bound}
\operatorname{kl}(a,b)\le\frac{(a-b)^2}{b(1-b)},
\qquad a\in[0,1],\quad b\in(0,1).
\end{equation}

For $x\in A_j$, Lemma~\ref{lem:cdf} gives
$F_0(x)=F_0(z_j)=H(z_j)\in[1/4,1/2]$. By Lemma~\ref{lem:perturbed-cdf},
$F_j(x)=F_0(x)+\varepsilon$. Equation~\eqref{eq:epsilon-range}
implies $\varepsilon\le1/8$, so $F_j(x)\in[1/4,5/8]$ and
$F_j(x)(1-F_j(x))\ge3/16$. Equation~\eqref{eq:bernoulli-kl-bound}
then gives $\operatorname{kl}(F_0(x),F_j(x))\le8\varepsilon^2$.
Outside $A_j$ the feedback laws coincide. Applying
Equation~\eqref{eq:adaptive-kl}, we obtain
\[
\operatorname{KL}(\mathbb P_0,\mathbb P_j)
\le8\varepsilon^2\mathbb E_0\!\left[
\sum_{t=1}^T\mathbf1\{x_t\in A_j\}\right]
=8\varepsilon^2\mathbb E_0N_j(T).
\]
Equation~\eqref{eq:bounded-expectations} with $W=N_j(T)\in[0,T]$
therefore gives
\[
\mathbb E_jN_j(T)
\le\mathbb E_0N_j(T)
+T\sqrt{\tfrac12\operatorname{KL}(\mathbb P_0,\mathbb P_j)}
\le\mathbb E_0N_j(T)+2\varepsilon T\sqrt{\mathbb E_0N_j(T)},
\]
proving the claim.
\end{proof}

\end{document}